\documentclass{article}
\usepackage[utf8]{inputenc}
\usepackage{amsfonts}
\usepackage{url,hyperref}
\usepackage{amsmath, amsthm, amssymb, subcaption}
\usepackage[capitalize]{cleveref}
\usepackage{thmtools}
\usepackage{thm-restate}
\usepackage{xcolor}
\usepackage[maxbibnames=99]{biblatex}
\usepackage{todonotes}
\usepackage{soul}
\usepackage{complexity}
\usepackage{xspace}
\usepackage{authblk}
\usepackage{comment}
\usepackage{tikz}
\usetikzlibrary{decorations.pathmorphing}

\definecolor{mypurple}{RGB}{208,134,255}
\definecolor{myblue}{RGB}{10,120,253}
\definecolor{myred}{RGB}{205,7,7}
\definecolor{mygreen}{RGB}{0,153,0}
\definecolor{myorange}{RGB}{244,154,33}
\definecolor{myyellow}{RGB}{225,225,2}
\definecolor{mydarkpurple}{RGB}{127,0,225}

\declaretheorem[name=Theorem]{theorem}
\declaretheorem[name=Lemma, sibling=theorem]{lemma}

\declaretheorem[name=Corollary, sibling=theorem]{corollary}

\declaretheorem[name=Definition, sibling=theorem]{definition}
\declaretheorem[name=Observation, sibling=theorem]{observation}

\DeclareMathOperator{\Left}{Left}
\def\S{\mathcal{S}} 

\newcommand{\stsep}{$s$-$t$-separator\xspace}

\newcommand{\stseps}{$s$-$t$-separators\xspace}
\newcommand{\stpath}{$s$-$t$-path\xspace}
\newcommand{\stpaths}{$s$-$t$-paths\xspace}
\newcommand{\stwalk}{$s$-$t$-walk\xspace}

\title{Minimum Separator Reconfiguration}
\author[1]{Guilherme~C.~M.~Gomes}
\author[2]{Clément~Legrand-Duchesne}
\author[3]{Reem~Mahmoud}
\author[4]{Amer~E.~Mouawad}
\author[5]{Yoshio~Okamoto}
\author[1]{Vinicius~F.~dos~Santos}
\author[6]{Tom~C.~van~der~Zanden}

\affil[1]{Department of Computer Science, Federal University of Minas Gerais, Belo Horizonte, Brazil\thanks{\{gcm.gomes,viniciussantos\}@dcc.ufmg.br}}
\affil[2]{LaBRI, CNRS, Université de Bordeaux, Bordeaux, France.\thanks{clement.legrand@labri.fr}}
\affil[3]{Virginia Commonwealth University, Richmond, VA, USA\thanks{mahmoudr@vcu.edu}}
\affil[4]{Department of Computer Science, American University of Beirut, Beirut, Lebanon\thanks{aa368@aub.edu.lb}}
\affil[5]{Graduate School of Informatics and Engineering, The University of Electro-Communications, Chofu, Japan\thanks{okamotoy@uec.ac.jp}}
\affil[6]{Department of Data Analytics and Digitalisation, Maastricht University, Maastricht, The Netherlands\thanks{T.vanderZanden@maastrichtuniversity.nl}}

\date{}

\newtheorem{question}{Open Problem}

\begin{document}

\maketitle

\begin{abstract}
We study the problem of reconfiguring one minimum \stsep $A$ into another minimum \stsep $B$ in some $n$-vertex graph $G$ containing two non-adjacent vertices $s$ and $t$. We consider several variants of the problem as we focus on both the token sliding and token jumping  models. Our first contribution is a polynomial-time algorithm that computes (if one exists) a minimum-length sequence of slides transforming $A$ into $B$. We additionally establish that the existence of a sequence of jumps (which need not be of minimum length) can be decided in polynomial time (by an algorithm that also outputs a witnessing sequence when one exists). In contrast, and somewhat surprisingly, we show that deciding if a sequence of at most $\ell$ jumps can transform $A$ into $B$ is an $\textsf{NP}$-complete problem. To complement this negative result, we investigate the parameterized complexity of what we believe to be the two most natural parameterized counterparts of the latter problem; in particular, we study the problem of computing a minimum-length sequence of jumps when parameterized by the size $k$ of the minimum \stseps and when parameterized by the number of jumps $\ell$. For the first parameterization, we show that the problem is fixed-parameter tractable, but does not admit a polynomial kernel unless $\textsf{NP} \subseteq \textsf{coNP/poly}$. We complete the picture by designing a kernel with $\mathcal{O}(\ell^2)$ vertices and edges for the length $\ell$ of the sequence as a parameter. 
\end{abstract}

\section{Introduction}
We study the problem of computing reconfiguration sequences between minimum \stseps. A set $S$ of vertices in a graph $G$ is an \emph{\stsep} if vertices $s$ and $t$ are separated in $G - S$, i.e, $s$ and $t$ belong to different components of $G - S$. A \emph{minimum \stsep} is an \stsep of minimum size. We always let $k$ denote the size of a minimum \stsep in $G$. The token jumping (TJ-) (resp. token sliding (TS-)) \textsc{Minimum Separator Reconfiguration (MSR)} problem is defined as follows. Given a graph $G$ and minimum \stseps $A$ and $B$, the goal is to determine if there exists a sequence of sets $A=S_1,S_2,\dots,S_r=B$, such that $S_i$ is a minimum \stsep, $S_i:=(S_{i-1}\setminus\{v\})\cup\{u\}$ for some $v\in S_{i-1}$, and $u\in V(G)\setminus S_{i-1}$ (resp. $u\in N_{G - S_{i-1}}(v)$) for every $i\in[r] \setminus \{1\}$.

\paragraph{Motivation.} 
Reconfiguration problems arise in various applications and, as a result, have gained considerable attention in recent literature \cite{dominating-set-recon,independent-set-recon,coloring-recon,perfect-matching-recon}. They appear in power supply problems, such as operating switches in a network to transform between different arrangements of power supply from stations to homes without causing a blackout \cite{power-supply}. They also show up in evolutionary biology, such as in the transformation of genomes via mutations \cite{genome-rearrangement}. Moreover, reconfiguration problems contribute to numerous fields of study, such as computational geometry with polygon reconfiguration \cite{polygons}, or statistical physics with the transformation of a particle's spin system \cite{statistical-physics}.
At the same time, vertex separators are useful in the factorization of sparse matrices \cite{matrix-factorization}, as well as, partitioning hypergraphs \cite{hypergraph-partitioning}. They also lend themselves to problems in cyber security and telecommunication \cite{telecom}, bioinformatics and computational biology \cite{bioinformatics}, and many divide-and-conquer graph algorithms \cite{vertex-sep}. Given the importance of vertex separators, we believe that it is a natural question to study the problem of reconfiguration between different vertex separators.

\paragraph{Related work.} Gomes, Nogueira, and dos Santos \cite{separator-recon} initiated the study of the problem of computing reconfiguration sequences between \stseps, $A$ and $B$, without restricting the size of the separators (to minimum). We call the corresponding problem \textsc{Vertex Separator Reconfiguration (VSR)}. They show that for token sliding, checking if $A$ can be transformed to $B$, i.e., \textsc{Vertex Separator Reconfiguration}, is a $\textsf{PSPACE}$-complete problem even on bipartite graphs. In contrast, under the token jumping model the problem becomes $\textsf{NP}$-complete for bipartite graphs. 

\paragraph{Our results.} Unlike the \textsc{VSR} problem, the requirement in the \textsc{MSR} problem that the separators in the reconfiguration sequence must be minimum introduces a lot of structure. In particular, we can rely on the duality between minimum separators and disjoint paths, observing that tokens are always constrained to move on a set of disjoint \stpaths, which we call \textit{canonical paths}. Using this property, we prove that, in an (optimal) solution, tokens always move ``forward'' towards their target locations and we never need to take a step back. This immediately prevents the problems from being \textsf{PSPACE}-complete since this gives a (polynomial) bound on the length of a solution. In fact, the ``always-forward'' property immediately implies a greedy algorithm that decides whether we can reconfigure one \stsep into another or not for both the token sliding and token jumping models. We then turn our attention to finding shortests reconfiguration sequences. While \textsc{TS-MSR} is still solvable in polynomial-time, finding an optimal solution for the \textsc{TJ-MSR} problem is shown to be \NP-complete by a reduction from \textsc{Vertex Cover}; finding the largest set of vertices that can be ``skipped'' by jumping over them is ``similar'' to finding a minimum vertex cover. 

We give a complete characterization of the (parameterized) complexity of the \textsc{TJ-MSR} problem for its natural parameterizations. In particular, we complement our \NP-hardness result by showing that the problem of finding a shortest sequence of token jumps is fixed-parameter tractable when parameterized by $k$, the size of a minimum separator; this is accomplished by further exploiting the structure imposed by the separators' minimality as yes-instances have pathwidth bounded by $\mathcal{O}{k}$. Unfortunately, unless $\textsf{NP} \subseteq \textsf{coNP/poly}$, the problem admits no polynomial kernel under this parameterization. Finally, we show that if we parameterize the problem by the length of the reconfiguration sequence, $\ell$, then we obtain a kernel with $\mathcal{O}(\ell^2)$ vertices and edges. 

\paragraph{Open problems.} It remains an interesting open question whether the jumping variant of the \textsc{VSR} problem is always in $\textsf{NP}$ or whether the problem is $\textsf{PSPACE}$-complete on general graphs. Recall that, when restricted to bipartite graphs, \textsc{TS-VSR} is $\textsf{PSPACE}$-complete while \textsc{TJ-VSR} is $\textsf{NP}$-complete~\cite{separator-recon}. We note that this rather intriguing disparity in complexity between the two models is an artifact of the complexity of the \textsc{Vertex Cover Reconfiguration} problem on bipartite graphs~\cite{DBLP:journals/talg/LokshtanovM19}; vertex covers in a bipartite graph 
$G = (L \cup R, E)$ correspond (one-to-one) to \stseps in the graph $G'$ obtained from $G$ by adding a vertex $s$ adjacent to every vertex in $L$ and a vertex $t$ adjacent to every vertex in $R$. 

\begin{question}
Is the \textsc{Vertex Separator Reconfiguration} problem under the token jumping model in $\textsf{NP}$? 
\end{question}

From the viewpoint of parameterized complexity, it is known that the \textsc{Vertex Cover Reconfiguration} problem (under both the jumping and sliding models) is fixed-parameter tractable when parameterized by the vertex cover size~\cite{DBLP:journals/algorithmica/MouawadN0SS17}. Unfortunately, this does not imply any positive result for the reconfiguration problem of (non-minimum) \stseps. However, it was shown in~\cite{DBLP:journals/algorithms/MouawadNRS18} that \textsc{Vertex Cover Reconfiguration} is $\textsf{W[1]}$-hard when parameterized by $\ell$, the length of a reconfiguration sequence, even when restricted to bipartite graphs. The model that the reduction is based on is the so-called token addition/removal model but the reduction can be adapted for the jumping and sliding models. We provide such adaptations in the appendix, which imply $\textsf{W[1]}$-hardness (of \textsc{Vertex Cover Reconfiguration} parameterized by $\ell$ and restricted to bipartite graphs) under the token jumping and token sliding models\footnote{There are two variants of the problem that one can consider under token sliding, i.e., either we allow multiple tokens to occupy the same vertex or not.} (Appendix~\ref{appendix-hardness}). Therefore, the reconfiguration problem for non-minimum \stseps, i.e., \textsc{Vertex Separator Reconfiguration}, is also $\textsf{W[1]}$-hard when parameterized by $\ell$ (under both models). We conclude this section with what we believe is the most relevant open question in this direction.

\begin{question}
Is the \textsc{Vertex Separator Reconfiguration} problem (under either the token jumping or token sliding model) fixed-parameter tractable when parameterized by the size of the separators?
\end{question}

We note that if we parameterize by both the size of the separators and the length of a reconfiguration sequence then, if we additionally assume that $A$ and $B$ are minimal, the problem of computing a sequence of minimal separators transforming $A$ to $B$ is fixed-parameter tractable (for both token jumps and token slides). The algorithm consists of first applying the treewidth reduction theorem of~\cite{DBLP:journals/talg/MarxOR13} followed by a reduction to the model checking problem in order to exploit Courcelle's dynamic programming machinery for graphs of bounded treewidth~\cite{DBLP:journals/iandc/Courcelle90}. We omit all the details and simply observe that the treewidth reduction theorem outputs a graph (of bounded treewidth) that preserves both the vertices participating in minimal \stseps as well as the structure of the graph induced by those vertices (which is required for token sliding). Once we obtain the aforementioned graph, we can convert the reconfiguration problem into a model checking problem by creating a monadic second-order formula that existentially quantifies over the required number of \stseps and verifies that all the necessary conditions (e.g., jumping, sliding, separation) are satisfied by the sequence of sets.

\section{Preliminaries}

We denote the set of natural numbers by $\mathbb{N}$ and, for $n \in \mathbb{N}$, we let $[n] = \{1, 2, \dots, n\}$.

\paragraph{Graphs.} We assume that each graph $G$ is finite, simple, and undirected.
We let~$V(G)$ and $E(G)$ denote the vertex set and edge set of $G$, respectively. 
The \emph{open neighborhood} of a vertex $v$ is denoted by $N_G(v) = \{u \in V(G) \mid \{u,v\} \in E(G)\}$ and the \emph{closed neighborhood} by $N_G[v] = N_G(v) \cup \{v\}$. 
For a set $S \subseteq V(G)$ of vertices, we define $N_G(S) = \bigcup_{v \in S} N_G(v) \setminus S$ and $N_G[S] = N_G(S) \cup S$. 
The subgraph of~$G$ \emph{induced} by $S$ is denoted by $G[S]$, where $G[S]$ has vertex set~$S$ and edge set $\{\{u,v\} \in E(G) \mid u,v \in S\}$. 
We let $G - S = G[V(G) \setminus S]$.

A \emph{walk} of length $q$ from $v_0$ to $v_q$ in $G$ is a vertex sequence $v_0, \dots, v_q$ such that $\{v_i,v_{i + 1}\} \in E(G)$ for all $i \in \{0, \dots, q-1\}$.
It is a \emph{path} if all vertices are distinct. 
An \emph{\stpath} is one with endpoints $s$ and $t$. Two \stpaths $P_1$ and $P_2$ are \emph{(internally) disjoint} if $V(P_1)\cap V(P_2) = \{s,t\}$. 
The following is a celebrated theorem attributed to Menger~\cite{MengerZurAK} and later generalized and made algorithmic by Ford and Fulkerson~\cite{10.5555/1942094}.

\begin{theorem}
The size of a minimum \stsep is equal to the maximum number of pairwise internally disjoint \stpaths. Moreover, a maximum set of pairwise internally disjoint \stpaths can be computed in polynomial time. 
\end{theorem}

\paragraph{Parameterized complexity.} 
A \emph{parameterized problem} $Q$ is a subset of $\Sigma^* \times \mathbb{N}$, where the second component denotes the \emph{parameter}. A parameterized problem is \emph{fixed-parameter tractable} with respect to a parameter $\kappa$, $\textsf{FPT}$ for short, if there exists an algorithm to decide whether $(x, \kappa) \in Q$ in time $f(\kappa) \cdot |x|^{\mathcal{O}(1)}$, where $f$ is a computable function. 

We say that two instances are \emph{equivalent} if they are both yes-instances or both no-instances. A \emph{kernelization algorithm}, or a \emph{kernelization} for short, is a polynomial-time algorithm that reduces an input instance $(x, \kappa)$ into an equivalent instance $(x', \kappa')$ such that $|x'|, \kappa' \leq f(\kappa)$, for some computable function $f$. Such $x'$ is called a \emph{kernel}. Every fixed-parameter tractable problem admits a kernel, however, possibly of exponential or worse size. For efficient algorithms it is therefore most desirable to obtain kernels of polynomial, or even linear size.
 
The \emph{$\textsf{W}$-hierarchy} is a collection of parameterized complexity classes $\textsf{FPT} \subseteq \textsf{W[1]} \subseteq \ldots \subseteq \textsf{W[t]}$. The conjecture $\textsf{FPT} \subsetneq \textsf{W[1]}$ can be seen as an analogue of the conjecture that $\textsf{P} \subsetneq \textsf{NP}$. Therefore, showing hardness in the parameterized setting is usually accomplished by establishing an $\textsf{FPT}$-reduction from a $\textsf{W}$-hard problem. We refer to the textbooks~\cite{DBLP:books/sp/CyganFKLMPPS15,DBLP:series/mcs/DowneyF99} for extensive background on parameterized complexity.

\section{Preprocessing and general observations}

We let $(G,s,t,A,B)$ denote an instance of the \textsc{ Minimum Separator Reconfiguration} problem, where $A$ and $B$ are minimum \stseps of size $k$. The model, i.e., jumping vs.\ sliding, will be clear from the context. We begin by making some general observations (that hold for both the jumping and sliding models) about the structure of sequences of minimum \stseps which we make extensive use of. We also describe some preprocessing operations: we assume they have been applied on any instance in the rest of the paper. We begin by introducing the notion of canonical paths, which describe the possible locations for each token.

\begin{definition}[Canonical paths]
Let $A$ and $B$ denote the starting and target separators, respectively. We begin by fixing a maximum set of pairwise internally disjoint \stpaths, which has size $k=|A|$ (since $A$ is a minimum separator). We may assume that all paths are \emph{chordless}, i.e., if two vertices of the same path are adjacent the edge connecting them is part of the path; if a path is not chordless we can decrease the length of the path by shortcutting along the chord. We repeat this procedure until all paths are chordless; this terminates since in each iteration the total number of vertices involved in the paths decreases. We call these $k$ chordless pairwise internally disjoint \stpaths $P_1,\ldots,P_k$ the \emph{canonical paths}.
\end{definition}

Note that canonical paths might not be uniquely defined; there may be multiple ways to choose $k$ chordless pairwise internally disjoint \stpaths. It suffices to fix any set of such paths as canonical.

\begin{lemma}\label{lemma:sep}
Let $S$ be a minimum \stsep. Then, $S$ contains exactly one vertex of each canonical path.
\end{lemma}
\begin{proof}
The set $S$ has to contain at least one vertex of each path, since otherwise there would be an \stpath in $G - S$. Since the number of paths is equal to the size of a minimum separator and the paths are vertex disjoint, it has to be exactly one of each.
\end{proof}

The next observation follows immediately from Lemma~\ref{lemma:sep}.

\begin{observation}\label{obs:moves}
For both token sliding and token jumping, each token is confined to its respective canonical path and in the case of sliding, a token can only slide to either one of its two neighbors along its canonical path.
\end{observation}




Thus, our view of the problem is that we are sliding (resp.\ jumping) tokens along a set of $k$ paths and that each token is confined to its respective path. We now show that we can always slide (resp. jump) a token in the direction of the target separator $B$ and never have to do a ``backward'' move.

Let $L(i)$ denote the number of vertices on the canonical path $P_i$, including $s$ and $t$.
Let $u_{i,1}, \dots u_{i, L(i)}$ denote the vertices on the canonical path $P_i$ in the order in which they appear on it, with $u_{i,1} = s$ and $u_{i, L(i)} = t$.
Let $a_i$ and $b_i$ denote the indices such that $V(P_i) \cap A = \{u_{i,a_i}\}$ and $V(P_i) \cap B = \{u_{i,b_i}\}$, i.e., $a_i$ is the index of the starting vertex of the token on $P_i$ and $b_i$ the index of the goal vertex for this token. Let $l_i = \min(a_i,b_i)$ and $r_i= \max(a_i,b_i)$. We first show that, in any (shortest) reconfiguration sequence, we only need to consider configurations of tokens in which, for all $i$, the token on the path $P_i$ remains between (or on) $u_{i,l_i}$ and $u_{i,r_i}$.

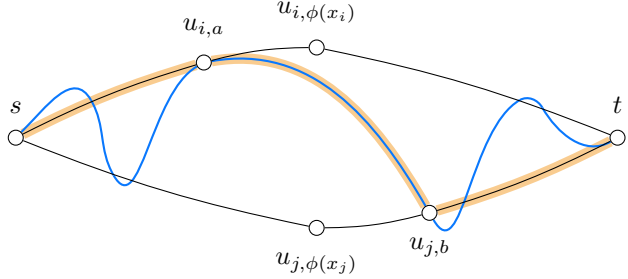
\begin{figure}[h!]
  \begin{center}
      \pgfdeclarelayer{background}
      \pgfsetlayers{background,main}
              
    \begin{tikzpicture}[decoration={snake,pre length=1mm,post length=1mm,segment length=4mm}]
      \node[draw=black,fill=white,circle,inner sep = 2pt] (s) at (0,0) {};
      \node[above=.1cm of s] {$s$};
      \node[draw=black,fill=white,circle,inner sep = 2pt] (t) at (8,0) {};
      \node[above=.1cm of t] {$t$};
      \node[draw=black,fill=white,circle,inner sep = 2pt] (uia) at (2.5,1) {};
      \node[above=.1cm of uia] {$u_{i,a}$};
      \node[draw=black,fill=white,circle,inner sep = 2pt] (uipxi) at (4,1.2) {};
      \node[above=.1cm of uipxi] {$u_{i,\phi(x_i)}$};
      \node[draw=black,fill=white,circle,inner sep = 2pt] (ujpxj) at (4,-1.2) {};
      \node[below=.1cm of ujpxj] {$u_{j,\phi(x_j)}$};
      \node[draw=black,fill=white,circle,inner sep = 2pt] (ujb) at (5.5,-1) {};
      \node[below=.1cm of ujb] {$u_{j,b}$};

      \draw[bend left=5] (s) edge (uia);
      \draw[bend left=7] (uia) edge (uipxi);
      \draw[bend left=5] (uipxi) edge (t);

      \draw[bend right=5] (s) edge (ujpxj);
      \draw[bend right=7] (ujpxj) edge (ujb);
      \draw[bend right=5] (ujb) edge (t);

      \begin{pgfonlayer}{background}
        \draw[myorange, opacity=0.5, line width=4pt] (s) edge[bend left=5] (uia) (uia) edge[out=10, in=120] (ujb) (ujb) edge[bend right=5] (t);

        \draw[myblue, thick] (s) edge[out=45, in=120, looseness=1.2] (1,0.5) (1,0.5) edge[out=300, in=120] (1.3,-0.5) (1.3,-0.5) edge[out=300, in=210] (uia) (uia) edge[out=10, in=120] (ujb) (ujb) edge[out=300, in=210] (6.7,0.5) (6.7,0.5) edge[out=30, in=120] (7,0.4) (7,0.4) edge[out=300, in=210] (t);
      \end{pgfonlayer}
    \end{tikzpicture}
  \end{center}
  \caption{$P$ is in blue and $P'$ is highlighted in orange.}
  \label{fig:phi-func}
\end{figure}  

\begin{lemma}
  For all $i \in [k]$, let $\phi_i$ be the function such that for all $1 < a < L(i)$,
  \[\phi_i(a) := \left\{
    \begin{array}{ll}
      l_i &\text{if } a <l_i,\\
      r_i &\text{if } a >r_i,\\
      a &\text{otherwise}.
    \end{array}\right.
  \]
  Let $f(u_{i,a}):= u_{i,\phi_i(a)}$. The image by $f$ of a minimum
  \stsep is a minimum \stsep.
\end{lemma}

\begin{proof}
  Given a set $X = \{u_{i,x_i} \mid i \in [k]\}$, let
  $\Left(X) = \bigcup_{i} \{u_{i,a} \mid a < x_i\}$.

  Let $X = \{u_{i,x_i} \mid i \in [k]\}$ be an \stsep. Assume that
  $Y = f(X)$ is not an \stsep. Let $P$ be an \stpath in
  $G - Y$. 

  Let $u_{i,a}$ be the last vertex of $P$ belonging to $\Left(Y)$. Let $u_{j,b}$ be
  the first vertex of $P$ lying on a canonical path after $u_{i,a}$. We have
  $u_{j,b} \notin \Left(Y)$ and thus $b > \phi_j(x_j)$. Note that $a \neq \phi_i(x_i)$ and $b \neq \phi_j(x_j)$ by definition of $P$. Let $P'$ be the path in $G$ going
  from $s$ to $u_{i,a}$ via $P_i$, then to $u_{j,b}$ via $P$ and finally to $t$
  via $P_j$; see Figure~\ref{fig:phi-func}.

  We have $i \neq j$ since otherwise $P'$ would be an \stpath in $G - A$
  (respectively $G - B$ or $G - X$) if $\phi_i(x_i) = a_i$
  (respectively $\phi_i(x_i) = b_i$ or $\phi_i(x_i) = x_i$).

  We cannot have $a < x_i $ and $x_j < b$ at the same time since otherwise $P'$ would be an \stpath
  in $G - X$. Without loss of generality, assume that $x_i \le a$. As a
  result, $x_i \le a < \phi_i(x_i) = l_i$. We must have $b < l_j$ since otherwise $P'$
  would be an \stpath in $G - A$ or $G - B$. Thus, $b < l_j \le \phi_j(x_j)$ 
  which contradicts
  the aformentioned property that $b > \phi_j(x_j)$
  and proves that $f(X)$ is an \stsep.
\end{proof}

\begin{corollary}\label{cor:bounded_region}
  In both the token jumping and token sliding models, if there exists a  reconfiguration sequence from $A$ to $B$, then there exists a shortest sequence such that, for any $i$, the $i^\textrm{th}$ token remains between $l_i$ and $r_i$ at all times.
  As a result, deleting all vertices on canonical paths that are not beteween $l_i$ and $r_i$ and replacing them by the edges $\{\{s,l_i\} \mid 1 \leq i \leq k\} \cup \{\{r_i,t\} \mid 1 \leq i \leq k\}$ yields an equivalent instance.
\end{corollary}

\begin{proof}
  Two \stseps differing only by a token jump (resp. slide) are mapped by $f$ to two \stseps that are either equal or differing by only a token jump (resp. slide). Thus, applying $f$ to all separators in the reconfiguration sequence, we get a reconfiguration sequence with the claimed property.
\end{proof}

Given a vertex $u_{i,x}$ on a canonical path $P_i$, we define the set
$F(u_{i,x})$ as the set of vertices of $P_i$ between $u_{i,x}$ and $u_{i,b_i}$ (see \cref{fig:F}).
We say that a jump (resp. slide) from $u_{i,a}$ to $u_{i,b}$ is \emph{forward} if
$u_{i,b} \in F(u_{i,a})$. Intuitively, this means that a jump (resp. slide) is forward if it moves the token closer to its target location (along the canonical path) without going past it.

\begin{figure}[h!]
  \begin{center}
      \pgfdeclarelayer{background}
      \pgfsetlayers{background,main}
              
    \begin{tikzpicture}[decoration={snake,pre length=1mm,post length=1mm,segment length=4mm}]
      \node[draw=black,fill=white,circle,inner sep = 2pt] (s) at (0,0) {};
      \node[above=.1cm of s] {$s$};
      \node[draw=black,fill=white,circle,inner sep = 2pt] (t) at (8,0) {};
      \node[above=.1cm of t] {$t$};
      \node[draw=black,fill=white,circle,inner sep = 2pt] (uisi) at (2,.9) {};
      \node[above=.1cm of uisi] {$u_{i,a_i}$};
      \node[draw=black,fill=white,circle,inner sep = 2pt] (uia) at (4,1.2) {};
      \node[above=.1cm of uia] {$u_{i,x}$};
      \node[draw=black,fill=white,circle,inner sep = 2pt] (uiti) at (6,.9) {};
      \node[above=.1cm of uiti] {$u_{i,b_i}$};
      \node[draw=black,fill=white,circle,inner sep = 2pt] (ujtj) at (2,-.9) {};
      \node[below=.1cm of ujtj] {$u_{j,b_j}$};
      \node[draw=black,fill=white,circle,inner sep = 2pt] (ujb) at (4,-1.2) {};
      \node[below=.1cm of ujb] {$u_{j,y}$};
      \node[draw=black,fill=white,circle,inner sep = 2pt] (ujsj) at (6,-.9) {};
      \node[below=.1cm of ujsj] {$u_{j,a_j}$};

      \draw[bend left=5] (s) edge (uisi);
      \draw[bend left=7] (uisi) edge (uia);
      \draw[bend left=7] (uia) edge (uiti);
      \draw[bend left=5] (uiti) edge (t);

      \draw[bend right=5] (s) edge (ujtj);
      \draw[bend right=7] (ujtj) edge (ujb);
      \draw[bend right=7] (ujb) edge (ujsj);
      \draw[bend right=5] (ujsj) edge (t);

      \begin{pgfonlayer}{background}
        \draw[bend left=7,line width = 4pt,opacity = .5,myorange] (uia) edge (uiti);
        \draw[bend right=7,line width = 4pt,opacity = .5,myblue] (ujtj) edge (ujb);        
      \end{pgfonlayer}
    \end{tikzpicture}
  \end{center}
  \caption{$F(u_{i,x})$ is highlighted in orange and $F(u_{j,y})$ in blue.}
  \label{fig:F}
\end{figure}
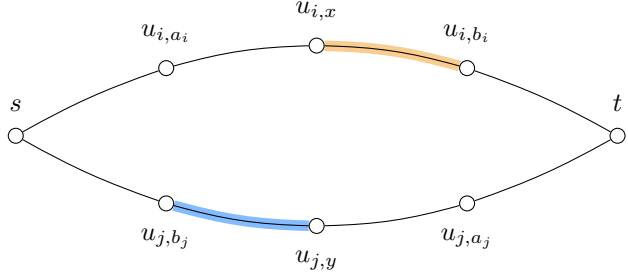  

\begin{lemma}[Forward-moving lemma]\label{lemma:only_forward}
If there exists a reconfiguration sequence from $A$ to $B$, then there exists a shortest sequence $\S$ of jumps (resp.\ slides) going from $A$ to $B$ containing only forward jumps (resp.\ slides).
\end{lemma}

\begin{proof}
  We proceed by induction on the length of a shortest sequence. We prove the case of jumps only, and the case of slides is similar.
  
  If $A = B$, there is nothing to prove.
  Otherwise, let $\S$ be a shortest sequence of jumps going from $A$ to $B$. By
  \cref{cor:bounded_region}, we can assume that the first jump of $\S$ is a
  forward jump. Let $A'$ be the separator obtained after this jump. The tail of $\S$
  is a shortest sequence of jumps between $A'$ and $B$ and by induction, it may be replaced with a shortest sequence that only contains forward jumps.
\end{proof}

\section{Polynomial-time algorithms}

The forward-moving lemma immediately implies that several problems can be solved in polynomial time by a greedy algorithm.

\begin{theorem}
A minimum-length sequence of token slides reconfiguring one minimum \stsep to another can be computed in polynomial time.
\end{theorem}

\begin{proof}
Since slides are reversible, doing any slide can never turn a yes-instance into a no-instance.
Thus, we can greedily apply moves that slide a token forward. Since we never need to do a backward slide, this always finds a solution if one exists. Moreover, this is optimal: the paths are chordless, so any slide can only advance a token one step towards its target position, and since there are no backward slides, the solution is optimal.
\end{proof}

\begin{theorem}\label{thm:tjdecide}
A (feasible, but not necessarily minimum-length) sequence of token jumps reconfiguring one minimum \stsep to another can be computed in polynomial time.
\end{theorem}

\begin{proof}
Jumps are also reversible, so doing a jump can never turn a yes-instance into a no-instance. Again, we can greedily apply forward jumps. Since a solution never needs to contain a backwards jump, this yields a feasible solution if one exists.
\end{proof}

In the case of token jumping, the solution produced by the greedy algorithm is not necessarily optimal (not guaranteed to be a shortest sequence of jumps); by choosing a different order for the jumps, it might be possible to make ``longer'' jumps, i.e.,  jumping over more vertices. In fact, we show that deciding whether a sequence of at most $\ell$ jumps can tranform one minimum \stsep into another is an $\textsf{NP}$-complete problem.

\section{Hardness of finding short sequences of jumps}
First, we note that the problem of deciding whether a sequence of at most $\ell$ jumps between two minimum \stseps exists is in $\textsf{NP}$. Indeed, Lemma~\ref{lemma:only_forward} implies that the length of a reconfiguration sequence cannot exceed $|V(G)|$; we can therefore directly use a reconfiguration sequence as a certificate. 

We show $\textsf{NP}$-hardness by reducing the \textsc{Vertex Cover} problem. Given a graph $G=(V,E)$ and the size $\kappa$ of a desired vertex cover, we construct our instance $(G',s,t,A,B)$ as follows. We first create a copy of the graph $G$, and for every $v\in V(G)$ we add two additional vertices $s_v,t_v$ and edges $\{s_v,v\},\{v,t_v\}$. We further add vertices $s,t$ and for all $v\in V(G)$, edges $\{s,s_v\},\{t_v,t\}$; see Figure~\ref{fig:np}. We ask whether we can reconfigure the \stsep $A = \{s_v\mid v\in V(G)\}$ to the separator $B = \{t_v\mid v\in V(G)\}$ using at most $|V(G)| + \kappa$ token jumps. 

Note that the canonical paths are of the form $\{s,s_v,v,s_t,t \mid v\in V(G)\}$.

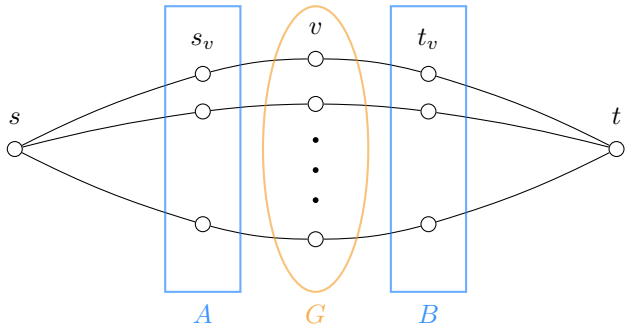
\begin{figure}[h!]
  \begin{center}
      \pgfdeclarelayer{background}
      \pgfsetlayers{background,main}
              
    \begin{tikzpicture}[decoration={snake,pre length=1mm,post length=1mm,segment length=4mm}]
      \node[draw=black,fill=white,circle,inner sep = 2pt] (s) at (0,0) {};
      \node[above=.1cm of s] {$s$};
      \node[draw=black,fill=white,circle,inner sep = 2pt] (t) at (8,0) {};
      \node[above=.1cm of t] {$t$};

      \draw (2.5,1) node[draw=black,fill=white,circle,inner sep = 2pt] (sv1) {} (4,1.2) node[draw=black,fill=white,circle,inner sep = 2pt] (v1) {} (5.5,1) node[draw=black,fill=white,circle,inner sep = 2pt] (tv1) {};
      \node[above=.1cm of sv1] {$s_v$}; \node[above=.1cm of v1] {$v$}; \node[above=.1cm of tv1] {$t_v$};

      \draw (2.5,0.5) node[draw=black,fill=white,circle,inner sep = 2pt] (sv2) {} (4,0.6) node[draw=black,fill=white,circle,inner sep = 2pt] (v2) {} (5.5,0.5) node[draw=black,fill=white,circle,inner sep = 2pt] (tv2) {};
      \draw (2.5,-1) node[draw=black,fill=white,circle,inner sep = 2pt] (sv3) {} (4,-1.2) node[draw=black,fill=white,circle,inner sep = 2pt] (v3) {} (5.5,-1) node[draw=black,fill=white,circle,inner sep = 2pt] (tv3) {};
      \draw (4,0.1) node {\huge{$\cdot$}} (4,-0.3) node {\huge{$\cdot$}} (4,-0.7) node {\huge{$\cdot$}};

      \foreach \i/\j/\k/\l in {1/left/7/5, 2/left/3/3, 3/right/7/5}
      {\draw[bend \j=\l] (s) edge (sv\i);
      \draw[bend \j=\k] (sv\i) edge (v\i);
      \draw[bend \j=\k] (v\i) edge (tv\i);
      \draw[bend \j=\l] (tv\i) edge (t);}

      \draw[thick, opacity=0.6, myorange] (4,0) ellipse (0.7 and 1.9);
      \draw[thick, opacity=0.6, myblue] (2,-1.9) rectangle (3,1.9) (5,-1.9) rectangle (6,1.9);
      \draw[opacity=0.8, myorange] (4,-2.2) node {$G$};
      \draw[opacity=0.8, myblue] (2.5,-2.2) node {$A$} (5.5,-2.2) node {$B$};
    \end{tikzpicture}
  \end{center}
  \caption{The graph $G'$ formed from $G$ (highlighted in orange), along with the initial and target separators $A$ and $B$ (highlighted in blue), respectively.}
  \label{fig:np}
\end{figure}  

\begin{lemma}\label{lem-hard-1}
If there exists a reconfiguration sequence from $A$ to $B$ in $G'$ consisting of at most $|V(G)|+\kappa$ token jumps, then $G$ has a vertex cover of size at most $\kappa$.
\end{lemma}

\begin{proof}
Since $A$ and $B$ are disjoint, every token must jump at least once, implying that we must use at least $|V(G)|$ jumps if every token jumps directly to its destination. Now let $\{u,v\} \in E(G)$. We claim that at some point during the reconfiguration sequence there should be a token on either $u$ or $v$. Otherwise, both tokens on $s_v$ and $s_u$ must jump directly to their destinations (since they are confined to move along their respective canonical paths). Without loss of generality, assume that the token on $s_u$ jumps to $t_u$ first. Then, there is an \stpath $s,s_u,u,v,t_v,v$; so, the configuration of tokens is not a separator and we have a contradiction. Therefore, for each edge in $G$, at least one of its endpoints requires that its corresponding token in $A$ performs an extra jump to an intermediate vertex in $G'$, i.e., to that endpoint (before jumping to its final destination). Since at most $\kappa$ tokens can jump twice, $G$ must have a vertex cover of size at most $\kappa$. 
\end{proof}

\begin{lemma}\label{lem-hard-2}
If $G$ has a vertex cover of size at most $\kappa$, then there exists a reconfiguration sequence from $A$ to $B$ in $G'$ consisting of at most $|V(G)| + \kappa$ token jumps.
\end{lemma}

\begin{proof}
Let $C$ be a vertex cover of $G$ of size at most $\kappa$. We first jump all tokens on vertices $\{s_v\mid v\in C\}$ to their corresponding intermediate vertices $\{v\mid v\in C\}$. Next, for all vertices $v \in V(G) \setminus C$, the tokens at $s_v$ directly jump to their destinations, followed by the remaining tokens jumping from the intermediate vertices to their destinations. Since tokens always stay on their corresponding canonical paths, any \stpath must use at least the vertices $s,s_v,v,u,t_u,t$ for some $u,v\in V(G)$. Suppose (towards a contradiction) that at some point there are no tokens on any of the vertices $s_v,v,u,t_u$. This means that there must be a token on $t_v$ (since there is no token on $s_v,v$) and on $s_u$ (since there is no token on $t_u$ nor on $u$). However, due to the way the sequence of jumps was defined, whenever there is a token on some vertex $t_v$ and a token on some vertex $s_u$, there are tokens forming a vertex cover of the subgraph induced by $V(G)$. This implies that there is a token on at least one of $u$ or $v$, which is a contradiction.
\end{proof}

\begin{theorem}\label{thm-jumping-nphard}
Deciding whether a sequence of at most $\ell$ token jumps can tranform one minimum \stsep into another is an $\textsf{NP}$-complete problem.
\end{theorem}

\begin{proof}
The proof follows from the fact that the problem is both $\textsf{NP}$-hard (Lemmas~\ref{lem-hard-1} and~\ref{lem-hard-2}) and in $\textsf{NP}$ (Lemma~\ref{lemma:only_forward}).
\end{proof}

\section{Preprocessing for token jumping}

In this section, we describe some preprocessing rules that can be applied to the token jumping variant of {\sc Minimum Separator Reconfiguration}. In the remainder of this paper, we assume that all instances are preprocessed according to these rules. We first show that we can reduce the graph so that it contains no vertices which are not on the canonical \stpaths.

\begin{lemma}\label{lemma:preprocesspaths}
Given an instance $(G,s,t,A,B)$ of {\sc Minimum Separator Reconfiguration} in the token jumping model, it is possible to compute in polynomial time an equivalent instance $(G',s,t,A,B)$ in which all vertices are on the (chordless) canonical paths and the minimum \stsep size is preserved. Moreover, all vertices in $A \cup B$ are adjacent to either $s$ or $t$.
\end{lemma}

\begin{proof}
 We define a graph $G'$ on a subset of the vertex set of $G$; in the new instance of {\sc Minimum Separator Reconfiguration} that we construct $s$, $t$, $A$, and $B$ stay the same.
 
By Corollary~\ref{cor:bounded_region}, we can assume that in an (optimal) solution, the token on canonical path $P_i$ only moves between $u_{i,l_i}$ and $u_{i,r_i}$. We build the graph $G'$ as follows. We take vertices $s$ and $t$ and, for each $i$, we take the vertices of the canonical path $P_i$ between (and including) $u_{i,l_i}$ and $u_{i,r_i}$ and add them to $G'$.
%
Now, whenever two non-adjacent vertices $u,v$ of $G'$ are connected by a path in $G$ and its internal vertices (i.e., the vertices which are not the endpoints of the path) are disjoint from the vertices of $G'$, we connect $u$ and $v$ by an edge in $G'$. Note that in particular this implies that we add the edges of the canonical paths; see Figure~\ref{fig:preprocessing}.

\begin{figure}[h!]
  \begin{center}
      \pgfdeclarelayer{background}
      \pgfsetlayers{background,main}
              
    \begin{tikzpicture}[decoration={snake,pre length=1mm,post length=1mm,segment length=4mm}]
      \node[draw=black,fill=white,circle,inner sep = 2pt] (s) at (0,0) {};
      \node[above=.1cm of s] {$s$};
      \node[draw=black,fill=white,circle,inner sep = 2pt] (t) at (8,0) {};
      \node[above=.1cm of t] {$t$};

      \draw (2.5,1) node[draw=black,fill=white,circle,inner sep = 2pt] (ul1) {} (4,1.2) node[draw=black,fill=white,circle,inner sep = 2pt] (v1) {} (5.5,1) node[draw=black,fill=white,circle,inner sep = 2pt] (ur1) {};
      \node[above=.1cm of ul1] {$u_{1,l_1}$}; \node[above=.1cm of ur1] {$u_{1,r_1}$};

      \draw (2.5,0.5) node[draw=black,fill=white,circle,inner sep = 2pt] (ul2) {} (4,0.6) node[draw=black,fill=white,circle,inner sep = 2pt] (v2) {} (5.5,0.5) node[draw=black,fill=white,circle,inner sep = 2pt] (ur2) {};
      \node[below=.1cm of ul2] {$u_{2,l_2}$}; \node[below=.1cm of ur2] {$u_{2,r_2}$};
      
      \draw (2.5,-1) node[draw=black,fill=white,circle,inner sep = 2pt] (ul3) {} (3.5,-1.2) node[draw=black,fill=white,circle,inner sep = 2pt] (v3) {} (5.5,-1) node[draw=black,fill=white,circle,inner sep = 2pt] (ur3) {};
      \node[below=.1cm of ul3] {$u_{k,l_k}$}; \node[below=.1cm of ur3] {$u_{k,r_k}$};
      \draw (4,0.1) node {\huge{$\cdot$}} (4,-0.3) node {\huge{$\cdot$}} (4,-0.7) node {\huge{$\cdot$}};

      \foreach \i/\j/\k/\l in {1/left/7/5, 2/left/3/3, 3/right/7/5}
      {\draw[bend \j=\l] (s) edge (ul\i);
      \draw[bend \j=\k] (ul\i) edge (v\i);
      \draw[bend \j=\k] (v\i) edge (ur\i);
      \draw[bend \j=\l] (ur\i) edge (t);}
      
      \node[draw=black, circle, fill=white, inner sep=2pt] at (3.2,0.57) {};
      \node[draw=black, circle, fill=white, inner sep=2pt] at (4.7,0.57) {};
      \node[draw=black, circle, fill=white, inner sep=2pt] at (4.5,-1.2) {};

      \node[draw=gray, circle, fill=gray!50!white, inner sep=1.5pt] at (1.1,0.5) {};
      \node[draw=gray, circle, fill=gray!50!white, inner sep=1.5pt] at (1.6,0.7) {};
      \node[draw=gray, circle, fill=gray!50!white, inner sep=1.5pt] at (6.9,-0.5) {};
      \node[draw=gray, circle, fill=gray!50!white, inner sep=1.5pt] at (6.4,-0.7) {};
      \node[draw=gray, circle, fill=gray!50!white, inner sep=1.5pt] at (1.1,-0.5) {};
      \node[draw=gray, circle, fill=gray!50!white, inner sep=1.5pt] at (6.9,0.5) {};
      \node[draw=gray, circle, fill=gray!50!white, inner sep=1.5pt] at (6.4,0.7) {};
      \node[draw=gray, circle, fill=gray!50!white, inner sep=1.5pt] at (6.6,0.3) {};
      \draw[draw=gray!80!white, thick] (ul3) edge[out=-40, in=-150, decorate, decoration={snake, amplitude=0.7mm}, looseness=1.2] (4.42,-1.27);
      \draw[draw=gray!80!white, thick] (3.3,0.6) edge[out=30, in=150, decorate, decoration={snake, amplitude=0.7mm}, looseness=1.2] (4.6,0.6);
    \end{tikzpicture}
\end{center}
\end{figure}

\begin{figure}[h!]
  \begin{center}
      \pgfdeclarelayer{background}
      \pgfsetlayers{background,main}
              
    \begin{tikzpicture}[decoration={snake,pre length=1mm,post length=1mm,segment length=4mm}]
      \node[draw=black,fill=white,circle,inner sep = 2pt] (s) at (0,0) {};
      \node[above=.1cm of s] {$s$};
      \node[draw=black,fill=white,circle,inner sep = 2pt] (t) at (8,0) {};
      \node[above=.1cm of t] {$t$};

      \draw (2.5,1) node[draw=black,fill=white,circle,inner sep = 2pt] (ul1) {} (4,1.2) node[draw=black,fill=white,circle,inner sep = 2pt] (v1) {} (5.5,1) node[draw=black,fill=white,circle,inner sep = 2pt] (ur1) {};
      \node[above=.1cm of ul1] {$u_{1,l_1}$}; \node[above=.1cm of ur1] {$u_{1,r_1}$};

      \draw (2.5,0.5) node[draw=black,fill=white,circle,inner sep = 2pt] (ul2) {} (4,0.6) node[draw=black,fill=white,circle,inner sep = 2pt] (v2) {} (5.5,0.5) node[draw=black,fill=white,circle,inner sep = 2pt] (ur2) {};
      \node[below=.1cm of ul2] {$u_{2,l_2}$}; \node[below=.1cm of ur2] {$u_{2,r_2}$};
      
      \draw (2.5,-1) node[draw=black,fill=white,circle,inner sep = 2pt] (ul3) {} (3.5,-1.2) node[draw=black,fill=white,circle,inner sep = 2pt] (v3) {} (5.5,-1) node[draw=black,fill=white,circle,inner sep = 2pt] (ur3) {};
      \node[below=.1cm of ul3] {$u_{k,l_k}$}; \node[below=.1cm of ur3] {$u_{k,r_k}$};
      \draw (4,0.1) node {\huge{$\cdot$}} (4,-0.3) node {\huge{$\cdot$}} (4,-0.7) node {\huge{$\cdot$}};

      \foreach \i/\j/\k/\l in {1/left/7/5, 2/left/3/3, 3/right/7/5}
      {\draw[bend \j=\l] (s) edge (ul\i);
      \draw[bend \j=\k] (ul\i) edge (v\i);
      \draw[bend \j=\k] (v\i) edge (ur\i);
      \draw[bend \j=\l] (ur\i) edge (t);}
      
      \node[draw=black, circle, fill=white, inner sep=2pt] at (3.2,0.57) {};
      \node[draw=black, circle, fill=white, inner sep=2pt] at (4.7,0.57) {};
      \node[draw=black, circle, fill=white, inner sep=2pt] at (4.5,-1.2) {};

      \draw (ul3) edge[bend right=45] (4.42,-1.27) (3.3,0.6) edge[bend left=45] (4.6,0.6);
    \end{tikzpicture}
  \end{center}
  \caption{Top: The graph $G$ (vertices and paths colored gray are deleted in $G'$). Bottom: The graph $G'$ formed from $G$ (gray paths in $G$ are replaced by edges in $G'$).}
  \label{fig:preprocessing}
\end{figure}
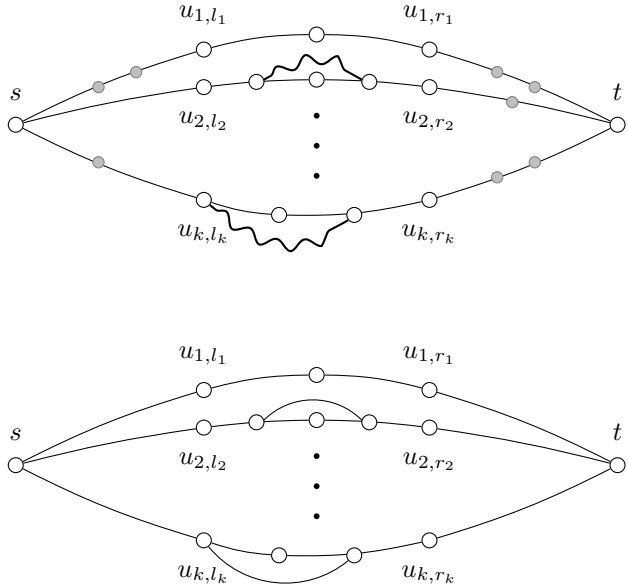


We claim that a configuration of tokens in which all tokens are between $u_{i,l_i}$ and $u_{i,r_i}$ on their  respective canonical paths $P_i$ is an \stsep in $G$ if and only if it is one in $G'$. Consider an \stpath $P$ in $G$. This path corresponds to an \stpath $P'$ in $G'$: if we consider only those vertices on the path that are in $G'$, that sequence of vertices is an \stpath in $G'$ since the consecutive vertices in that sequence are connected by a path in $G$ whose internal vertices are disjoint from $G'$ and thus there exist edges containing the consecutive vertices in the sequence. Thus, a configuration of tokens which is an \stsep in $G'$ is also an \stsep in $G$, since any configuration of tokens that hits $P'$ also hits $P$ since the vertices of $P'$ are a subset of the vertices of $P$.

Conversely, consider an \stpath $P'$ in $G'$: recall that the edges of $G'$ (which are not in $G$) correspond to paths in $G$ so the \stpath $P'$ corresponds to an \stwalk $W$ of $G$. Any configuration of tokens in which all tokens are between $u_{i,l_i}$ and $u_{i,r_i}$ on their  respective canonical paths that hits $W$ must do so in a vertex that is also in $P'$ since the vertices that are in $G'$ are precisely the vertices on canonical paths between $u_{i,l_i}$ and $u_{i,r_i}$. Thus, any configuration of tokens (in which all tokens are between $u_{i,l_i}$ and $u_{i,r_i}$ on their  respective canonical paths) that is an \stsep in $G$ is also an \stsep in $G'$.

Clearly, a token jump in $G$ (in which tokens remain between $u_{i,l_i}$ and $u_{i,r_i}$ on their respective canonical paths) is also a token jump in $G'$ and vice versa. Note that this does not hold under the token sliding model; adding an edge between two vertices of the same canonical path that are non-adjacent in $G$ (but connected by some paths) will allow for a slide in $G'$ that is not possible in $G$. The minimum separator size in $G'$ has not decreased since we still have the same number of disjoint paths (and also not increased since we have shown that the separators of $G$ are also separators of $G'$).

It is possible that this process creates chords on the canonical paths, making them no longer canonical. In this case, we can shorcut the canonical paths along the chords, and repeat the above process. Since the number of vertices decreases in each iteration, this eventually terminates. We eventually end up with the claimed, equivalent instance. This takes polynomial time.
\end{proof}

Going forward, we assume that all graphs are preprocessed according to these rules. We next observe that to ensure that a configuration of tokens forms a valid \stsep it suffices to check the existence of relatively simple \stpaths.

\begin{lemma}\label{lemma:crossonce}
To check whether a configuration of tokens that assigns exactly one token to each canonical path forms an \stsep, it suffices to check whether there exists an \stpath that from $s$, follows one of the canonical paths, then follows one edge (a \emph{crossing edge}) from that canonical path to another, and then follows that canonical path to $t$.
\end{lemma}

\begin{proof}
Suppose we have an \stpath $P$ of the form $$s,\ldots,u_{i,l},u_{j,m},\dots,u_{j,n},u_{q,o},\dots,t$$ where $i\neq j$, $j\neq q$, i.e., an \stpath that has more than one crossing edge. Let $P'$ be the \stpath that goes from $s$ along the canonical path $P_j$ to $u_{j,m}$ and then follows $P$ to $t$. Let $P''$ be the \stpath that follows $P$ from $s$ to $u_{j,m}$ and then runs along the canonical path $P_j$ to $t$; see Figure~\ref{fig:crossing-edges} (left). If some configuration of tokens (in which exactly one token is on each canonical path) does not hit $P$, then it also does not hit either $P'$ or $P''$: there is only one token on the path $P_j$, so it must be either before $u_{j,m}$ or after $u_{j,n}$. Therefore, it suffices to check only $P'$ and $P''$ and both $P'$ and $P''$ have fewer crossing edges than $P$. Repeatedly applying this procedure, we see that it is sufficient to check only those \stpaths with at most one crossing edge. Since exactly one token is assigned to each canonical path, \stpaths with no crossing edges are automatically covered and thus it suffices to check only those \stpaths with exactly one crossing edge.
\end{proof}

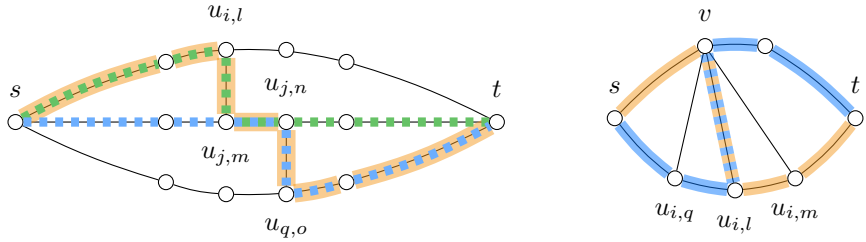
\begin{figure}[h!]
\begin{minipage}{0.6\textwidth}
  \begin{center}
      \pgfdeclarelayer{background}
      \pgfsetlayers{background,main}
              
    \begin{tikzpicture}[decoration={snake,pre length=1mm,post length=1mm,segment length=4mm}, scale=0.8]
      \node[draw=black,fill=white,circle,inner sep = 2pt] (s) at (0,0) {};
      \node[above=.1cm of s] {$s$};
      \node[draw=black,fill=white,circle,inner sep = 2pt] (t) at (8,0) {};
      \node[above=.1cm of t] {$t$};

      \draw (2.5,1) node[draw=black,fill=white,circle,inner sep = 2pt] (ul1) {} (3.5,1.2) node[draw=black,fill=white,circle,inner sep = 2pt] (v1) {} (5.5,1) node[draw=black,fill=white,circle,inner sep = 2pt] (ur1) {};

      \draw (2.5,0) node[draw=black,fill=white,circle,inner sep = 2pt] (ul2) {} (3.5,0) node[draw=black,fill=white,circle,inner sep = 2pt] (v2) {} (5.5,0) node[draw=black,fill=white,circle,inner sep = 2pt] (ur2) {};
      
      \draw (2.5,-1) node[draw=black,fill=white,circle,inner sep = 2pt] (ul3) {} (3.5,-1.2) node[draw=black,fill=white,circle,inner sep = 2pt] (v3) {} (5.5,-1) node[draw=black,fill=white,circle,inner sep = 2pt] (ur3) {};

      \foreach \i/\j/\k/\l in {1/left/7/5, 2/left/0/0, 3/right/7/5}
      {\draw[bend \j=\l] (s) edge (ul\i);
      \draw[bend \j=\k] (ul\i) edge (v\i);
      \draw[bend \j=\k] (v\i) edge (ur\i);
      \draw[bend \j=\l] (ur\i) edge (t);}
      
      \node[draw=black, circle, fill=white, inner sep=2pt] (w1) at (4.5,1.2) {};
      \node[draw=black, circle, fill=white, inner sep=2pt] (w2) at (4.5,0) {};
      \node[draw=black, circle, fill=white, inner sep=2pt] (w3) at (4.5,-1.2) {};
      \draw (v1) -- (v2) (w2) -- (w3);

      \node[above=.1cm of v1] {$u_{i,l}$}; 
      \node[below=.1cm of v2] {$u_{j,m}$}; \node[above=.1cm of w2] {$u_{j,n}$};
      \node[below=.1cm of w3] {$u_{q,o}$};

      \draw[opacity=0.5, line width=7pt, myorange] (s) edge[bend left=7] (ul1) (ul1) edge[bend left=5] (v1) (v1) -- (v2) -- (w2) -- (w3) (w3) edge[bend right=5] (ur3) (ur3) edge[bend right=7] (t);
      \draw[line width=3.5pt, mygreen!60!white, dashed] (s) edge[bend left=7] (ul1) (ul1) edge[bend left=5] (v1) (v1) -- (v2) (w2) -- (ur2) -- (t);
      \draw[line width=3.5pt, myblue!60!white, dashed] (s) -- (ul2) -- (v2) (w2) -- (w3) (w3) edge[bend right=5] (ur3) (ur3) edge[bend right=7] (t);
      \draw[mygreen!60!white,dash pattern= on 3pt, dash phase=3pt, line width=3.5pt] (v2) -- (w2);
      \draw[myblue!60!white,dash pattern= on 3pt, line width=3.5pt] (v2) -- (w2);
    \end{tikzpicture}
  \end{center}
\end{minipage}%
\begin{minipage}{0.45\textwidth}
  \begin{center}
      \pgfdeclarelayer{background}
      \pgfsetlayers{background,main}
              
    \begin{tikzpicture}[decoration={snake,pre length=1mm,post length=1mm,segment length=4mm}, scale=0.8]
      \node[draw=black,fill=white,circle,inner sep = 2pt] (s) at (0,0) {};
      \node[above=.1cm of s] {$s$};
      \node[draw=black,fill=white,circle,inner sep = 2pt] (t) at (4,0) {};
      \node[above=.1cm of t] {$t$};

      \draw (1.5,1.2) node[draw=black,fill=white,circle,inner sep = 2pt] (v1) {} (2.5,1.2) node[draw=black,fill=white,circle,inner sep = 2pt] (w1) {} (1,-1) node[draw=black,fill=white,circle,inner sep = 2pt] (ul3) {} (2,-1.2) node[draw=black,fill=white,circle,inner sep = 2pt] (v3) {} (3,-1) node[draw=black,fill=white,circle,inner sep = 2pt] (ur3) {};

      \draw[bend left=7] (s) edge (v1);
      \draw[bend left=5] (v1) edge (w1);
      \draw[bend left=7] (w1) edge (t);
      \draw[bend right=7] (s) edge (ul3);
      \draw[bend right=5] (ul3) edge (v3);
      \draw[bend right=5] (v3) edge (ur3);
      \draw[bend right=7] (ur3) edge (t);
      \draw (ul3) -- (v1) -- (v3) (v1) -- (ur3);
      
      \node[above=.1cm of v1] {$v$}; 
      \node[below=.1cm of ul3] {$u_{i,q}$}; \node[below=.1cm of v3] {$u_{i,l}$};
      \node[below=.1cm of ur3] {$u_{i,m}$};

      \draw[opacity=0.5, line width=5pt, myorange] (s) edge[bend left=7] (v1) (v3) edge[bend right=5] (ur3) (ur3) edge[bend right=7] (t);
      \draw[opacity=0.5, line width=5pt, myblue] (s) edge[bend right=7] (ul3) (ul3) edge[bend right=5] (v3) (v1) edge[bend left=5] (w1) (w1) edge[bend left=7] (t);
      \draw[opacity=0.5, line width=5pt, myblue, dash pattern= on 3pt] (v1) -- (v3);
      \draw[opacity=0.5, line width=5pt, myorange, dash pattern= on 3pt, dash phase= 3pt] (v1) -- (v3);
    \end{tikzpicture}
  \end{center}
  \end{minipage}
  \caption{Left: An example of Lemma~\ref{lemma:crossonce}. The path $P$ is highlighted in orange, $P'$ is the dashed blue path, and $P''$ is the dashed green path. Note that $P,P',P''$ all overlap between $u_{j,m}$ and $u_{j,n}$. Right: An example of Lemma~\ref{lemma:degthree}. If the \stpath highlighted in orange (resp. blue) is not hit by a token, replacing the edge $\{v,u_{i,l}\}$ with $\{v,u_{i,m}\}$ (resp. $\{v,u_{i,q}\}$) in that path gives another \stpath not hit by a token.}
  \label{fig:crossing-edges}
\end{figure}

We now show that it suffices to consider instances in which all vertices have degree at least $3$ and at most $2k$, where $k \geq 3$ is the size of the minimum separator (for $k \leq 2$, the problem can be solved in polynomial time by a simple algorithm that computes a shortest path in an auxiliary graph $H$ having one vertex for each of the at most $n^2$ \stseps and where two vertices of $H$ share an edge whenever the corresponding \stseps are one reconfiguration step away from each other). 

\begin{lemma}\label{lemma:degthree}
Given an instance $(G,s,t,A,B)$ of {\sc Minimum Separator Reconfiguration} in the token jumping model, it is possible to compute in polynomial time an equivalent instance $(G',s,t,A,B)$ in which all vertices have degree at least $3$ and at most $2k$, where $k=|A|=|B|$ (assuming $k \geq 3$). In particular, every vertex can have at most two neighbors on each canonical path. 
\end{lemma}

\begin{proof}
Consider a vertex $v$ that is adjacent to three distinct vertices $u_{i,q}, u_{i,l}, u_{i,m}$ that belong to the same canonical path (different from the canonical path of $v$) and where $q<l<m$. Then, the edge $\{v,u_{i,l}\}$ can be deleted; if there is an \stpath that uses $\{v,u_{i,l}\}$ as a crossing edge and this path is not hit by some configuration of tokens, then, depending on the location of the token on path $P_i$, that configuration of tokens fails to hit either the \stpath that uses $\{v,u_{i,q}\}$ as a crossing edge or the \stpath that uses $\{v,u_{i,m}\}$ as a crossing edge; see Figure~\ref{fig:crossing-edges} (right). Thus, we can assume that each vertex is adjacent to at most $2$ vertices on each canonical path, and thus, each vertex has degree at most $2k$. We can also assume, without loss of generality, that there are no vertices of degree $0$ or $1$ and that any vertex $v$ of degree exactly two is deleted and replaced by an edge connecting the two neighbors of $v$ on its canonical path. 
\end{proof}

We proceed by showing that we can always assume that the source and target minimum  \stseps, i.e., $A$ and $B$, are disjoint. 

\begin{lemma}\label{lemma:disjoint}
Given an instance $(G,s,t,A,B)$ of {\sc Minimum Separator Reconfiguration} in the token jumping model, it is possible to compute in polynomial time an equivalent instance $(G',s,t,A',B')$ in which $A' \cap B' = \emptyset$.
\end{lemma}

\begin{proof}
Let $v$ be a vertex in both $A$ and $B$. By the forward-moving lemma, i.e., Lemma~\ref{lemma:only_forward}, we know that in any shortest reconfiguration sequence from $A$ to $B$ the token on $v$ does not move. 
Hence, we can simply delete the vertex $v$ from the graph (reducing the minimum separator size by one) and we adjust $A$ and $B$ accordingly (by deleting $v$ from each) to obtain an equivalent instance $(G -\{v\},s,t,A \setminus \{v\},B \setminus \{v\})$. We repeat this procedure until the source and target \stseps are disjoint. Since each iteration takes polynomial time and we can repeat the process at most $k$ times, the whole procedure takes polynomial time.  
\end{proof}

We conclude this section by formalizing the notion of unskippable vertices; a notion that will be useful in many of our subsequent results (and that was already implicity used in Lemma~\ref{lem-hard-1}). Given an instance $(G,s,t,A,B)$, we say that $v \in V(G) \setminus (A \cup B \cup \{s,t\})$ is \emph{unskippable} if for every sequence (if any exist) of minimum \stseps that transforms $A$ to $B$, there exists at least one \stsep $S$ in the sequence such that $v \in S$. In other words, there is no transformation from $A$ to $B$ that can skip over $v$ and not jump a token onto $v$ at some point. 
Similarly, we say that a set $U \subseteq V(G) \setminus (A \cup B \cup \{s,t\})$ of vertices is \emph{unskippable} whenever there exists at least one \stsep $S$ in every reconfiguration sequence (from $A$ to $B$) such that $|U \cap S| \geq 1$.

\begin{lemma}\label{lemma:unskip}
Let $(G,s,t,A,B)$ be an instance of {\sc Minimum Separator Reconfiguration} in the token jumping model. 
A vertex $v \in V(G) \setminus (A \cup B \cup \{s,t\})$ having two neighbors on a canonical path other than its own is unskippable. If $u,v \in  V(G) \setminus (A \cup B \cup \{s,t\})$, $u,v$ belong to two different canonical paths, and $\{u,v\} \in E(G)$, then $\{u,v\}$ is unskippable.
\end{lemma}

\begin{proof}
Let $v$ be a vertex on some canonical path $P$ having two neighbors on canonical path $P' \neq P$. Let $u_1$ and $u_2$ denote the neighbors of $v$ on $P'$ such that $u_1$ appears before $u_2$ on $P'$ ($u_1$ is closer to $s$ on $P'$). Assume that there exists a sequence of jumps transforming $A$ to $B$ such that the token on $P$ never jumps to vertex $v$. Consider the first separator $S$ of the sequence in which the token on $P$ is after $v$ (on some vertex of $P$ that is closer to $t$ along $P$). If the token on $P'$ is before $u_2$ then we have an \stpath using the crossing edge $\{v,u_2\}$; see Figure~\ref{fig:skippable-case1}. Hence, the token on $P'$ must be on or after $u_2$. But in that case, consider the separator $S'$ that occurs right before $S$ in the sequence. In $S'$, the token on $P$ is before $v$ and the token on $P'$ is on or after $u_2$. We again get a contradiction as we can now construct an \stpath using the crossing edge $\{v,u_1\}$. We therefore conclude that $v$ is unskippable. 

Now, consider $u,v \in  V(G) \setminus (A \cup B \cup \{s,t\})$ such that $u$ and $v$ belong to different canonical paths, say $P_u$ and $P_v$, and $\{u,v\} \in E(G)$. Assume that there exists a sequence transforming $A$ to $B$ such that no token ever jumps on neither $u$ nor $v$. Then, there must exist a first \stsep in the transformation such that either the token on $P_u$ is after $u$ and the token on $P_v$ is before $v$ or vice versa. In either case, we get an \stpath and the required contradiction; see Figure~\ref{fig:skippable-case2}. We then conclude that $\{u,v\}$ is unskippable. 
\end{proof}

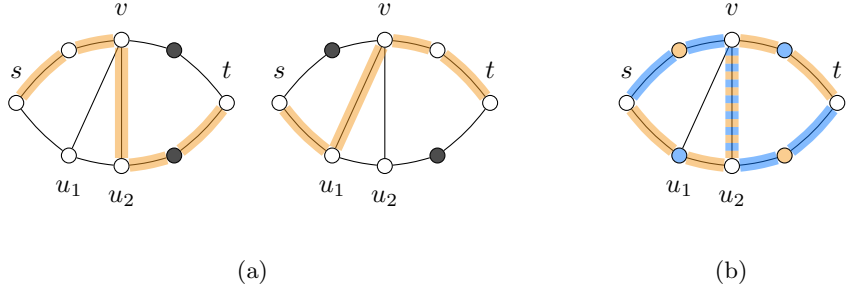
\begin{figure}[h!]
\begin{subfigure}{0.65\textwidth}
  \begin{center}
      \pgfdeclarelayer{background}
      \pgfsetlayers{background,main}
              
    \begin{tikzpicture}[decoration={snake,pre length=1mm,post length=1mm,segment length=4mm}, scale=0.7]
      \node[draw=black,fill=white,circle,inner sep = 2pt] (s) at (0,0) {};
      \node[above=.1cm of s] {$s$};
      \node[draw=black,fill=white,circle,inner sep = 2pt] (t) at (4,0) {};
      \node[above=.1cm of t] {$t$};

      \draw (1,1) node[draw=black,fill=white,circle,inner sep = 2pt] (ul1) {} (2,1.2) node[draw=black,fill=white,circle,inner sep = 2pt] (v1) {} (3,1) node[draw=black,fill=black!70!white,circle,inner sep = 2pt] (ur1) {} (1,-1) node[draw=black,fill=white,circle,inner sep = 2pt] (ul2) {} (2,-1.2) node[draw=black,fill=white,circle,inner sep = 2pt] (v2) {} (3,-1) node[draw=black,fill=black!70!white,circle,inner sep = 2pt] (ur2) {};

      \draw[bend left=7] (s) edge (ul1);
      \draw[bend left=5] (ul1) edge (v1);
      \draw[bend left=5] (v1) edge (ur1);
      \draw[bend left=7] (ur1) edge (t);
      \draw[bend right=7] (s) edge (ul2);
      \draw[bend right=5] (ul2) edge (v2);
      \draw[bend right=5] (v2) edge (ur2);
      \draw[bend right=7] (ur2) edge (t);
      \draw (ul2) -- (v1) -- (v2);
      
      \node[above=.1cm of v1] {$v$};
      \node[below=.1cm of ul2] {$u_1$};
      \node[below=.1cm of v2] {$u_2$};

      \draw[opacity=0.5, line width=5pt, myorange] (s) edge[bend left=7] (ul1) (ul1) edge[bend left=5] (v1) (v1) -- (v2) (v2) edge[bend right=5] (ur2) (ur2) edge[bend right=7] (t);

      \node[draw=black,fill=white,circle,inner sep = 2pt] (s) at (5,0) {};
      \node[above=.1cm of s] {$s$};
      \node[draw=black,fill=white,circle,inner sep = 2pt] (t) at (9,0) {};
      \node[above=.1cm of t] {$t$};

      \draw (6,1) node[draw=black,fill=black!70!white,circle,inner sep = 2pt] (ul1) {} (7,1.2) node[draw=black,fill=white,circle,inner sep = 2pt] (v1) {} (8,1) node[draw=black,fill=white,circle,inner sep = 2pt] (ur1) {} (6,-1) node[draw=black,fill=white,circle,inner sep = 2pt] (ul2) {} (7,-1.2) node[draw=black,fill=white,circle,inner sep = 2pt] (v2) {} (8,-1) node[draw=black,fill=black!70!white,circle,inner sep = 2pt] (ur2) {};

      \draw[bend left=7] (s) edge (ul1);
      \draw[bend left=5] (ul1) edge (v1);
      \draw[bend left=5] (v1) edge (ur1);
      \draw[bend left=7] (ur1) edge (t);
      \draw[bend right=7] (s) edge (ul2);
      \draw[bend right=5] (ul2) edge (v2);
      \draw[bend right=5] (v2) edge (ur2);
      \draw[bend right=7] (ur2) edge (t);
      \draw (ul2) -- (v1) -- (v2);
      
      \node[above=.1cm of v1] {$v$};
      \node[below=.1cm of ul2] {$u_1$};
      \node[below=.1cm of v2] {$u_2$};

      \draw[opacity=0.5, line width=5pt, myorange] (s) edge[bend right=7] (ul2) (ul2) -- (v1) (v1) edge[bend left=5] (ur1) (ur1) edge[bend left=7] (t);
    \end{tikzpicture}
  \end{center}
  \caption{}
  \label{fig:skippable-case1}
  \end{subfigure}%
  \begin{subfigure}{0.4\textwidth}
  \begin{center}
      \pgfdeclarelayer{background}
      \pgfsetlayers{background,main}
              
    \begin{tikzpicture}[decoration={snake,pre length=1mm,post length=1mm,segment length=4mm}, scale=0.7]
      \node[draw=black,fill=white,circle,inner sep = 2pt] (s) at (0,0) {};
      \node[above=.1cm of s] {$s$};
      \node[draw=black,fill=white,circle,inner sep = 2pt] (t) at (4,0) {};
      \node[above=.1cm of t] {$t$};

      \draw (1,1) node[draw=black,fill=myorange!50!white,circle,inner sep = 2pt] (ul1) {} (2,1.2) node[draw=black,fill=white,circle,inner sep = 2pt] (v1) {} (3,1) node[draw=black,fill=myblue!50!white,circle,inner sep = 2pt] (ur1) {} (1,-1) node[draw=black,fill=myblue!50!white,circle,inner sep = 2pt] (ul2) {} (2,-1.2) node[draw=black,fill=white,circle,inner sep = 2pt] (v2) {} (3,-1) node[draw=black,fill=myorange!50!white,circle,inner sep = 2pt] (ur2) {};

      \draw[bend left=7] (s) edge (ul1);
      \draw[bend left=5] (ul1) edge (v1);
      \draw[bend left=5] (v1) edge (ur1);
      \draw[bend left=7] (ur1) edge (t);
      \draw[bend right=7] (s) edge (ul2);
      \draw[bend right=5] (ul2) edge (v2);
      \draw[bend right=5] (v2) edge (ur2);
      \draw[bend right=7] (ur2) edge (t);
      \draw (ul2) -- (v1) -- (v2);
      
      \node[above=.1cm of v1] {$v$};
      \node[below=.1cm of ul2] {$u_1$};
      \node[below=.1cm of v2] {$u_2$};

      \draw[opacity=0.5, line width=5pt, myorange] (s) edge[bend right=7] (ul2) (ul2) edge[bend right=5] (v2) (v1) edge[bend left=5] (ur1) (ur1) edge[bend left=7] (t);
      \draw[opacity=0.5, line width=5pt, myblue] (s) edge[bend left=7] (ul1) (ul1) edge[bend left=5] (v1) (v2) edge[bend right=5] (ur2) (ur2) edge[bend right=7] (t);
      \draw[opacity=0.5, line width=5pt, myblue, dash pattern= on 3pt] (v1) -- (v2);
      \draw[opacity=0.5, line width=5pt, myorange, dash pattern= on 3pt, dash phase= 3pt] (v1) -- (v2);
    \end{tikzpicture}
  \end{center}
  \caption{}
  \label{fig:skippable-case2}
  \end{subfigure}
  \caption{(a) Black vertices represent tokens. Left: The position of tokens in $S$. The bottom token must come after $u_2$ in $S$ to hit the \stpath highlighted in orange. Right: The position of tokens in $S'$. The \stpath highlighted in orange is not hit by any token. (b) Orange vertices represent one possible token configuration and blue vertices represent another. For the orange (resp.\ blue) configuration, we get the \stpath in orange (resp.\ blue).}
\end{figure}   

\section{Fixed-parameter tractability for parameter $k$}

In this section, we show that finding a shortest sequence for the token jumping variant of {\sc Minimum Separator Reconfiguration} is fixed-parameter tractable with respect to the size $k$ of minimum \stseps. We do so by constructing a path decomposition of the graph from a not necessarily shortest  sequence of token jumps, and then proceed by dynamic programming on that path decomposition. We shall work towards proving the following.

\begin{theorem}\label{thm:mainFPT}
The optimization version of {\sc Minimum Separator Reconfiguration} in the token jumping model is fixed-parameter tractable with respect to the size $k$ of a minimum separator.
\end{theorem}

\begin{proof}
Recall that, by Lemma~\ref{lemma:preprocesspaths}, we can preprocess the graph so that all vertices are on the canonical paths and the vertices in $A$ and $B$ are adjacent to (at least one of) $s$ or $t$. By Lemma~\ref{lemma:degthree}, each vertex in $G$ has degree at least $3$ and at most $2k$ and, by Lemma~\ref{lemma:disjoint}, we know that $A \cap B = \emptyset$. 

For simplicity, we remove the vertices $s$ and $t$ from the graph and in this section view the problem as follows. Given two disjoint sets of vertices $A$ and $B$ which are connected by $k=|A|=|B|$ disjoint paths, the aim is to reconfigure the tokens from $A$ to $B$ under the constraint that at no point should a path from any vertex in $A$ to any vertex in $B$ open up (i.e., not be hit by some token).

First observe that this is a sufficient condition. This is because, as $A$ and $B$ are disjoint \stseps, any \stpath must pass through both $A$ and $B$ (and thus, includes a path from some vertex in $A$ to some vertex in $B$ as a subpath).

To see that this is also necessary, suppose some configuration of tokens (that is an \stsep) does not hit some path $P$ from some vertex $a\in A$ to some vertex $b\in B$. We can assume this path only intersects $A$ in $a$ and only intersects $B$ in $b$ since otherwise, we could take a subpath of $P$. Then, without loss of generality, both $a$ and $b$ should be adjacent to $s$ (they could also both be adjacent to $t$, but since the configuration of tokens is an \stsep, it cannot be the case that one of them is adjacent to $s$ while the other is to $t$). This implies $a$ and $b$ are on different canonical paths (due to the preprocessing and the chordless nature of the canonical paths, the vertices in $A$ and $B$ are the only vertices that can connect to $s$ and $t$, so if $a\in A$ and $b\in B$ are on the same canonical path, one of them must connect to $s$ and the other to $t$ since otherwise, the canonical path would not be an \stpath), and there is a path $P'$ (a subpath of $P$) connecting these canonical paths that does not pass through any vertex of $A$ or $B$. Let $a'$ be the vertex on the canonical path of $a$ that is in $B$. We can now find an \stpath that will contradict the assumption that $A$ is an \stsep: $t$ is adjacent to $a'$; we walk along the canonical path of $a'$ towards $a$ until we reach $P'$ (which must happen before we reach the token on $A$), then follow $P'$ to the canonical path of $b$, then walk along it towards $b$ and finally reach $s$. This \stpath is not hit by $A$ and we reach a contradiction: $A$ is not an \stsep. Therefore, we can conclude that any configuration of tokens that is an \stsep hits every path from any vertex in $A$ to any vertex in $B$.

Note that similarly to Lemma~\ref{lemma:crossonce}, it suffices to check for paths from some vertex in $A$ to some vertex in $B$ with exactly one crossing edge.

Because we have gotten rid of $s$ and $t$, we can now pick one \emph{consistent} ordering of the vertices of the canonical paths. Previously, some vertices of $A$ may have been adjacent to $t$ and some others adjacent to $s$, which means that the order of vertices as they appear on the paths from $A$ to $B$ might not have been the same order in which they appear from $s$ to $t$. Going forward, we assume the vertices of the canonical paths are numbered from $A$ to $B$, i.e., $u_{i,j}$ is the $j^\textrm{th}$ vertex of the $i^\textrm{th}$ canonical path from $A$ to $B$, in particular $u_{i,1}\in A$ is the first (leftmost) vertex of the path, $u_{i,L(i)}\in B$ the last (rightmost) vertex of the path.

Since we only need to consider forward jumps, jumps are always from some vertex $u_{i,j}$ to another vertex $u_{i,q}$ with $q > j$. We say a token on vertex $u_{i,j}$ is before, on or after vertex $u_{i,q}$ if $j<q, j=q$, or $j>q$,  respectively.

\paragraph{Path decomposition.}
We now construct a path decomposition of $G$. We start by computing a (not necessarily shortest) token jumping solution using Theorem~\ref{thm:tjdecide}. By Lemma~\ref{lemma:only_forward}, we can assume this solution only contains forward jumps. The token jumping solution directly gives a path decomposition of $G$: if $S_i$ and $S_{i+1}$ are consecutive separators in the reconfiguration sequence, then they differ in the replacement of one vertex $u_{j,q}$ by one vertex $u_{j,l}$ which are both from the same canonical path. That is, $S_i=X\cup\{u_{j,q}\}, S_{i+1}=X\cup\{u_{j,l}\}$ for some $X$. We replace this jump by the sequence of bags $X\cup \{u_{j,q}\}, X\cup \{u_{j,q},u_{j,q+1}\}, X\cup \{u_{j,q+1},u_{j,q+2}\},\ldots, X\cup \{u_{j,l-1},u_{j,l}\}, X\cup \{u_{j,l}\}$. The largest bag in this decomposition contains $k+1$ vertices, and hence the width of this path decomposition is $k$.

Since the token jumping sequence moves all tokens from one endpoint of each canonical path to the other, we see that the bags of the decomposition contain all vertices of each path (and thus all vertices of the graph) and that (since there are only forward jumps) the bags containing each specific vertex induce a connected subpath of the decomposition. By construction, each edge of a canonical path is contained in the decomposition. Note that the canonical paths are chordless and therefore the only remaining edges we have to concern ourselves with are the edges connecting vertices on distinct canonical paths.

For an edge $\{u,v\}$ that connects two vertices $u,v$ that are on different canonical paths, recall that a token cannot move past $u$ before the other token has moved at least onto $v$ or vice versa (Lemma~\ref{lemma:unskip}). This implies that either $u,v$ occur together in some separator of the reconfiguration sequence (and hence $u,v$ appear together in some bag of the decomposition) or that at some point in the reconfiguration sequence, either $u$ or $v$ appears in some separator and then the other vertex is jumped over (we say a vertex $v$ is \emph{jumped over} if a token jumps from vertex $x$ to vertex $y$ and the vertex $v$ appears along the canonical path from $x$ to $y$). In this case, there is also a bag containing both $u$ and $v$.

Note that the path decomposition is constructed in such a way that
\begin{itemize}
    \item The first bag is equal to $A$;
    \item The last bag is equal to $B$;
    \item Vertices are introduced in the order they appear on the canonical paths (in particular, the jump from a previously introduced vertex to a newly introduced one would be a forward jump).
\end{itemize}

\paragraph{Dynamic programming.}

As is customary for dynamic programming on path decompositions, we process the bags of the decomposition in order, and for each bag $X$, we consider the graph induced by the vertices in $X$ and all vertices in bags to the left of $X$. We call this graph $G_X$. As usual, we assume the decomposition is given in \emph{nice} form, i.e., each bag differs from the previous either by the addition of one vertex (we say that such a vertex is \emph{introduced}) or the removal of one vertex (we say such a vertex is \emph{forgotten}). We refer the reader to~\cite{DBLP:journals/tcs/Bodlaender98, DBLP:journals/jgt/HarveyW17, DBLP:journals/jal/RobertsonS86} for an extensive background on pathwidth, (nice) path decompositions, and their applications. 


A \emph{partial solution} $\Pi$ is a sequence of configurations of tokens $C_1,\ldots, C_l$, starting with $A$, so that each configuration $C_i$ is obtained from the previous one by one of the following operations:
\begin{enumerate}
    \item removing a token from a canonical path and replacing it by another token further along the same canonical path (i.e., a forward jump);
    \item removing a token from a canonical path completely,
\end{enumerate}
under the condition that if $v$ is the rightmost vertex of a canonical path in $G_X$ (note that this implies $v\in X$), then $v$ may not be reachable from any vertex of $A$ in the subgraph $G_X - C_i$ unless there is no token on the canonical path of $v$ in $C_i$. Intuitively, $C_i$ should separate $A$ from $X$, except that vertices on canonical paths from which the token has been removed may be reachable (and, in the case where the bag contains two vertices from the same canonical path, only the rightmost one needs to be unreachable).

The operation (1) is just a normal (forward) jump operation of reconfiguration via token jumping, whereas the operation (2) is necessary to represent a potential future jump to a vertex that has not yet been introduced.

We now define the \emph{configuration characteristic} of a configuration of tokens $C$ with respect to $X$ as follows. For every vertex $v$ in the closed neighborhood of $X$, we store its \emph{state}: whether the token on the canonical path of $v$ is before, on, or after $v$.

Since by Lemma~\ref{lemma:degthree} every vertex has degree at most $2k$ there are at most $3^{2k^2+2k}$ distinct characteristics of token configurations.\footnote{If for a vertex $v\in X$ both of its neighbors on the canonical path are also present in $X$, the state of $v$ itself is implied and so the bound is valid.}


The \emph{solution characteristic} $\mathcal{C}$ of a partial solution $\Pi$ is now defined as the sequence of characteristics of the token configurations occurring in $\Pi$, with runs of consecutive duplicate configuration characteristics removed and replaced with one occurence of the configuration characteristic. Note that since tokens can only jump forward, we will never repeat a characteristic. Thus, the number of distinct characteristics of partial solutions is at most $(3^{2k^2+2k}+1)!$.

Our dynamic programming table will store for each solution characteristic the minimum number of configurations of a partial solution having that characteristic. We now show how to, given the dynamic programming table for the previous bag, the corresponding table for the next bag (which may be either an introduce or forget bag) may be computed.

\paragraph{Introduce.} Let $X'$ be the previous bag in the decomposition and let $X=X'\cup\{v\}$ be the next bag where vertex $v$ is being introduced. Consider a partial solution $\Pi$ for $G_{X}$. If $\Pi$ does not contain any jump to $v$, then $\Pi$ is also a partial solution for $G_{X'}$: $G_{X'}$ differs from $G_X$ only by the deletion of $v$; deleting $v$ cannot make any vertex reachable. Moreover, if $u$ is the previous vertex on the canonical path of $v$, then $u$ cannot have been reachable in any configuration in which the token on its canonical path had not been removed since then we could have reached $v$ from $u$.

Otherwise, if $\Pi$ contains a jump to $v$, then, because all jumps are forward, this is the last jump operation on vertex $v$ (though it may possibly be followed by the removal of the token on $v$). Consider $\Pi'$ that is obtained from $\Pi$ by replacing the last jump of the token to $v$ with a removal of the token (and ignoring the possible duplicate subsequent removal of the token on $v$). We claim that $\Pi'$ is a partial solution for $G_{X'}$: if in $\Pi'$, at any time a path opens up (in $G_{X'}$) from some vertex in $A$ to some vertex in $X'$ (excluding the canonical path of $v$), then, since $X'\subseteq X$, this path would also be an $A$-$X$-path in $\Pi$ (and $G_{X}$).

Thus, we can obtain all possible partial solutions for $G_{X}$ by considering the partial solutions for $G_{X'}$ and either
\begin{enumerate}
    \item doing nothing;
    \item replacing the removal of the token on the canonical path of $v$ with a jump to $v$; or
    \item replacing the removal of the token on the canonical path of $v$ with a jump to $v$ \emph{and} at some point in the future, adding a removal of (the token on) $v$.
\end{enumerate}

Intuitively, (1) corresponds to jumping over $v$, (3) corresponds to jumping onto it and later moving away and (2) is necessary to represent that the token on $v$ has reached its target location (i.e., $v\in B$) and will not make further jumps.

\begin{figure}[h!]
  \begin{center}
      \pgfdeclarelayer{background}
      \pgfsetlayers{background,main}
              
    \begin{tikzpicture}[decoration={snake,pre length=1mm,post length=1mm,segment length=4mm}, scale=0.8]

      \draw (0,0.5) -- (5,0.5) (0,-0.5) -- (4,-0.5);
      \draw (0,0) ellipse (0.8cm and 1.2cm) (3,0) ellipse (0.8cm and 1.2cm) (-1,-1.5) rectangle (4,1.5);
      \draw (0,0) node {$A$} (3,0) node {$X'$} (0,0.5) node[draw=black, circle, inner sep=1.2pt, fill=white] {} (0,-0.5) node[draw=black, circle, inner sep=1.2pt, fill=white] {} (3,0.5) node[draw=black, circle, inner sep=1.2pt, fill=white] {$u$} (3,-0.5) node[draw=black, circle, inner sep=1.2pt, fill=white] {} (4.5,0.5) node[draw=black, circle, inner sep=1.2pt, fill=white] {$v$} (1.5,-2) node {$G_{X'}$}; 
      
      \draw (5.8,0) edge[->] (7.8,0) (6.8,0.3) node {Introduce};
      
      \draw (10,0.5) -- (14.2,0.5) (10,-0.5) -- (14.2,-0.5);
      \draw (10,0) ellipse (0.8cm and 1.2cm) (13,0) ellipse (0.8cm and 1.2cm) (9,-1.5) rectangle (14,1.5);
      \draw (10,0) node {$A$} (13,0) node {$X$} (10,0.5) node[draw=black, circle, inner sep=1.2pt, fill=white] {} (10,-0.5) node[draw=black, circle, inner sep=1.2pt, fill=white] {} (12.65,0.5) node[draw=black, circle, inner sep=1.2pt, fill=white] {$u$} (13,-0.5) node[draw=black, circle, inner sep=1.2pt, fill=white] {} (13.35,0.5) node[draw=black, circle, inner sep=1.2pt, fill=white] {$v$} (11.5,-2) node {$G_X$};
    \end{tikzpicture}

    \vspace{1cm}

    \begin{tikzpicture}[decoration={snake,pre length=1mm,post length=1mm,segment length=4mm}, scale=0.8]

      \draw (0,0.5) -- (4.2,0.5) (0,-0.5) -- (4.2,-0.5);
      \draw (0,0) ellipse (0.8cm and 1.2cm) (3,0) ellipse (0.8cm and 1.2cm) (-1,-1.5) rectangle (4,1.5);
      \draw (0,0) node {$A$} (3,0) node {$X'$} (0,0.5) node[draw=black, circle, inner sep=1.2pt, fill=white] {} (0,-0.5) node[draw=black, circle, inner sep=1.2pt, fill=white] {} (2.65,0.5) node[draw=black, circle, inner sep=1.2pt, fill=white] {$u$} (3,-0.5) node[draw=black, circle, inner sep=1.2pt, fill=white] {} (3.35,0.5) node[draw=black, circle, inner sep=1.2pt, fill=white] {$v$} (1.5,-2) node {$G_{X'}$}; 
      
      \draw (5.8,0) edge[->] (7.8,0) (6.8,0.3) node {Forget};
      
      \draw (10,0.5) -- (14.2,0.5) (10,-0.5) -- (14.2,-0.5);
      \draw (10,0) ellipse (0.8cm and 1.2cm) (13,0) ellipse (0.8cm and 1.2cm) (9,-1.5) rectangle (14,1.5);
      \draw (10,0) node {$A$} (13,0) node {$X$} (10,0.5) node[draw=black, circle, inner sep=1.2pt, fill=white] {} (10,-0.5) node[draw=black, circle, inner sep=1.2pt, fill=white] {} (11.7,0.5) node[draw=black, circle, inner sep=1.2pt, fill=white] {$u$} (13,-0.5) node[draw=black, circle, inner sep=1.2pt, fill=white] {} (13,0.5) node[draw=black, circle, inner sep=1.2pt, fill=white] {$v$} (11.5,-2) node {$G_X$};
    \end{tikzpicture}
  \end{center}
  \caption{An overview of the process of introducing and forgetting a vertex.}
  \label{fig:introduce-forget}
\end{figure}
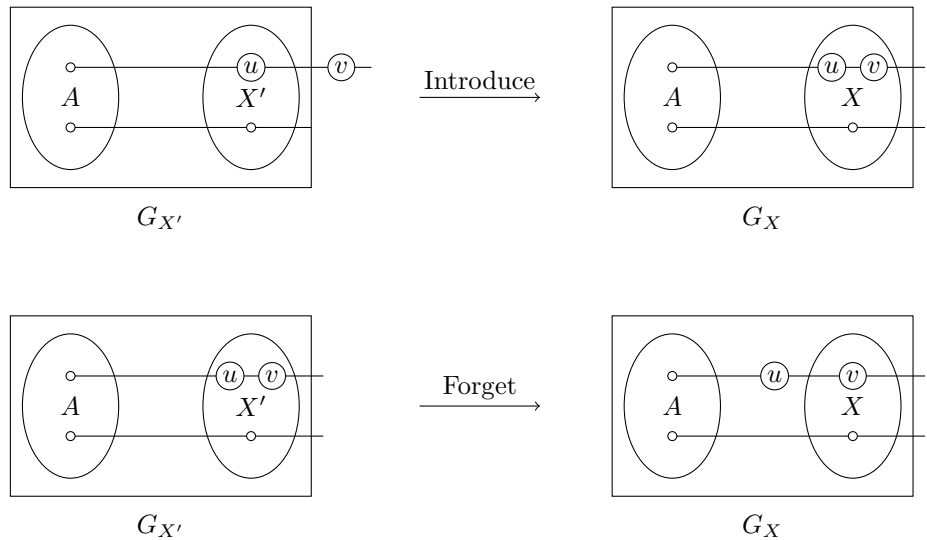  

\begin{lemma}
    Let $\mathcal{C}'$ be a characteristic of a partial solution $\Pi'$ with respect to $G_{X'}$. Then, the characteristic $\mathcal{C} = \mathcal{C}'$ and the legality of the partial solution $\Pi$ with respect to $G_{X}$ obtained from $\Pi'$ by doing nothing depends only on, and can be computed from, $\mathcal{C}'$.
\end{lemma}

\begin{proof}
Consider the characteristic for $\Pi'$ (with respect to $G_{X'}$). Since (by the basic properties of path decompositions) $v$ only has neighbors in $X'$, the only information in the characteristic for $\Pi$ (w.r.t. $G_X=G_{X\cup \{v\}}$) that we are missing is the state of $v$. However, the vertex preceding $v$ on its canonical path is also in $X$ (because the vertices are introduced in the order they appear on the canonical paths) and we can infer the state of $v$ from it: if its state is ``before'' or ``on'' then $v$ is ``before'', otherwise the token has been removed and $v$ is ``after''.

To check whether the partial solution is legal when we extend it by doing nothing, we have to check whether an $A$-$X$-path opens up at any point. We only have to worry about paths that use $v$ since other paths are already precluded by $\Pi'$ being a (legal) partial solution. For this it suffices to check that in no characteristic, $v$ gets state ``after'' while one of its neighbors still has state ``before.''
\end{proof}

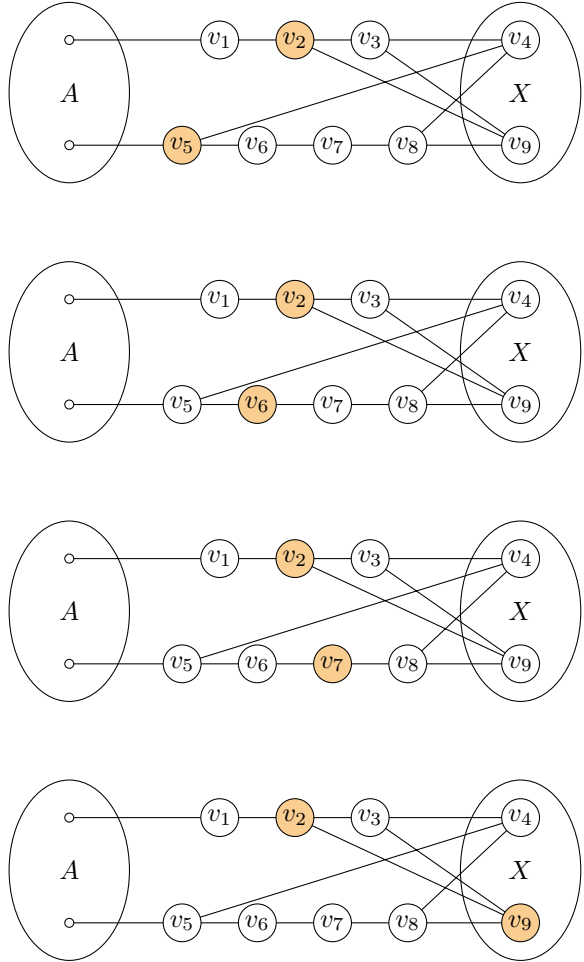
\begin{figure}[h!]
  \begin{center}
      \pgfdeclarelayer{background}
      \pgfsetlayers{background,main}
              
    \begin{tikzpicture}[decoration={snake,pre length=1mm,post length=1mm,segment length=4mm}]

      \draw (0,0.7) -- (6,0.7) (0,-0.7) -- (6,-0.7);
      \draw (3,0.7) -- (6,-0.7) -- (4,0.7) (1.5,-0.7) -- (6,0.7) -- (4.5,-0.7);
      \draw (0,0) ellipse (0.8cm and 1.2cm) (6,0) ellipse (0.8cm and 1.2cm);
      \draw (0,0) node {$A$} (6,0) node {$X$} (2,0.7) node[draw=black, circle, inner sep=1.2pt, fill=white] {$v_1$} (3,0.7) node[draw=black, circle, inner sep=1.2pt, fill=myorange!50!white] {$v_2$} (4,0.7) node[draw=black, circle, inner sep=1.2pt, fill=white] {$v_3$} (1.5,-0.7) node[draw=black, circle, inner sep=1.2pt, fill=myorange!50!white] {$v_5$} (2.5,-0.7) node[draw=black, circle, inner sep=1.2pt, fill=white] {$v_6$} (3.5,-0.7) node[draw=black, circle, inner sep=1.2pt, fill=white] {$v_7$} (4.5,-0.7) node[draw=black, circle, inner sep=1.2pt, fill=white] {$v_8$} (6,0.7) node[draw=black, circle, inner sep=1.2pt, fill=white] {$v_4$} (6,-0.7) node[draw=black, circle, inner sep=1.2pt, fill=white] {$v_9$};
      \draw (0,0.7) node[draw=black, circle, inner sep=1.2pt, fill=white] {} (0,-0.7) node[draw=black, circle, inner sep=1.2pt, fill=white] {};
    \end{tikzpicture}

    \vspace{1cm}

    \begin{tikzpicture}[decoration={snake,pre length=1mm,post length=1mm,segment length=4mm}]

      \draw (0,0.7) -- (6,0.7) (0,-0.7) -- (6,-0.7);
      \draw (3,0.7) -- (6,-0.7) -- (4,0.7) (1.5,-0.7) -- (6,0.7) -- (4.5,-0.7);
      \draw (0,0) ellipse (0.8cm and 1.2cm) (6,0) ellipse (0.8cm and 1.2cm);
      \draw (0,0) node {$A$} (6,0) node {$X$} (2,0.7) node[draw=black, circle, inner sep=1.2pt, fill=white] {$v_1$} (3,0.7) node[draw=black, circle, inner sep=1.2pt, fill=myorange!50!white] {$v_2$} (4,0.7) node[draw=black, circle, inner sep=1.2pt, fill=white] {$v_3$} (1.5,-0.7) node[draw=black, circle, inner sep=1.2pt, fill=white] {$v_5$} (2.5,-0.7) node[draw=black, circle, inner sep=1.2pt, fill=myorange!50!white] {$v_6$} (3.5,-0.7) node[draw=black, circle, inner sep=1.2pt, fill=white] {$v_7$} (4.5,-0.7) node[draw=black, circle, inner sep=1.2pt, fill=white] {$v_8$} (6,0.7) node[draw=black, circle, inner sep=1.2pt, fill=white] {$v_4$} (6,-0.7) node[draw=black, circle, inner sep=1.2pt, fill=white] {$v_9$};
      \draw (0,0.7) node[draw=black, circle, inner sep=1.2pt, fill=white] {} (0,-0.7) node[draw=black, circle, inner sep=1.2pt, fill=white] {};
    \end{tikzpicture}

    \vspace{1cm}

    \begin{tikzpicture}[decoration={snake,pre length=1mm,post length=1mm,segment length=4mm}]

      \draw (0,0.7) -- (6,0.7) (0,-0.7) -- (6,-0.7);
      \draw (3,0.7) -- (6,-0.7) -- (4,0.7) (1.5,-0.7) -- (6,0.7) -- (4.5,-0.7);
      \draw (0,0) ellipse (0.8cm and 1.2cm) (6,0) ellipse (0.8cm and 1.2cm);
      \draw (0,0) node {$A$} (6,0) node {$X$} (2,0.7) node[draw=black, circle, inner sep=1.2pt, fill=white] {$v_1$} (3,0.7) node[draw=black, circle, inner sep=1.2pt, fill=myorange!50!white] {$v_2$} (4,0.7) node[draw=black, circle, inner sep=1.2pt, fill=white] {$v_3$} (1.5,-0.7) node[draw=black, circle, inner sep=1.2pt, fill=white] {$v_5$} (2.5,-0.7) node[draw=black, circle, inner sep=1.2pt, fill=white] {$v_6$} (3.5,-0.7) node[draw=black, circle, inner sep=1.2pt, fill=myorange!50!white] {$v_7$} (4.5,-0.7) node[draw=black, circle, inner sep=1.2pt, fill=white] {$v_8$} (6,0.7) node[draw=black, circle, inner sep=1.2pt, fill=white] {$v_4$} (6,-0.7) node[draw=black, circle, inner sep=1.2pt, fill=white] {$v_9$};
      \draw (0,0.7) node[draw=black, circle, inner sep=1.2pt, fill=white] {} (0,-0.7) node[draw=black, circle, inner sep=1.2pt, fill=white] {};
    \end{tikzpicture}

    \vspace{1cm}

    \begin{tikzpicture}[decoration={snake,pre length=1mm,post length=1mm,segment length=4mm}]

      \draw (0,0.7) -- (6,0.7) (0,-0.7) -- (6,-0.7);
      \draw (3,0.7) -- (6,-0.7) -- (4,0.7) (1.5,-0.7) -- (6,0.7) -- (4.5,-0.7);
      \draw (0,0) ellipse (0.8cm and 1.2cm) (6,0) ellipse (0.8cm and 1.2cm);
      \draw (0,0) node {$A$} (6,0) node {$X$} (2,0.7) node[draw=black, circle, inner sep=1.2pt, fill=white] {$v_1$} (3,0.7) node[draw=black, circle, inner sep=1.2pt, fill=myorange!50!white] {$v_2$} (4,0.7) node[draw=black, circle, inner sep=1.2pt, fill=white] {$v_3$} (1.5,-0.7) node[draw=black, circle, inner sep=1.2pt, fill=white] {$v_5$} (2.5,-0.7) node[draw=black, circle, inner sep=1.2pt, fill=white] {$v_6$} (3.5,-0.7) node[draw=black, circle, inner sep=1.2pt, fill=white] {$v_7$} (4.5,-0.7) node[draw=black, circle, inner sep=1.2pt, fill=white] {$v_8$} (6,0.7) node[draw=black, circle, inner sep=1.2pt, fill=white] {$v_4$} (6,-0.7) node[draw=black, circle, inner sep=1.2pt, fill=myorange!50!white] {$v_9$};
      \draw (0,0.7) node[draw=black, circle, inner sep=1.2pt, fill=white] {} (0,-0.7) node[draw=black, circle, inner sep=1.2pt, fill=white] {};
    \end{tikzpicture}
  \end{center}
  \caption{The above sequence of token configurations (partial solution) has the following solution characteristic. 1\textsuperscript{st} (Top): $v_2=$ on, $v_3=$ before, $v_4=$ before, $v_5=$ at, $v_8=$ before, $v_9=$ before. 2\textsuperscript{nd} \& 3\textsuperscript{rd}: $v_2=$ on, $v_3=$ before, $v_4=$ before, $v_5=$ after, $v_8=$ before, $v_9=$ before. 4\textsuperscript{th} (Bottom): $v_2=$ on, $v_3=$ before, $v_4=$ before, $v_5=$ after, $v_8=$ after, $v_9=$ on. Since the 2\textsuperscript{nd} and 3\textsuperscript{rd} configuration characteristic are the same, the solution characteristic will consist of just three configuration characteristics. Note that this figure is only intended to illustrate the notion of solution characteristic and does not represent a legal partial solution, since $v_4$ becomes reachable when we jump from $v_5$ to $v_6$.}
  \label{fig:characteristic}
\end{figure} 

\begin{lemma}
   Let $\mathcal{C}'$ be a characteristic of a partial solution $\Pi'$ with respect to $G_{X'}$ that contains the removal of the token on the canonical path of $v$. Then, the characteristic $\mathcal{C}$, and the legality of the partial solution $\Pi$ with respect to $G_{X}$ obtained from $\Pi'$ by replacing the removal with a jump to $v$ depends only on, and can be computed from, $\mathcal{C}'$.
\end{lemma}

\begin{proof}
As we argued with the ``do nothing'' case, the only missing information is the state of $v$ itself. The state of $v$ can again be inferred from the state of its predecessor on the path (which is again in $X'$): it starts out at ``before'' and changes to ``on'' at the moment the state of its predecessor changes to ``after'' (this corresponds to the moment of the removal, which is replaced by a token jump to vertex $v$).

The characteristics again provide us with enough information that no $A$-$X$-path is created: similarly to the ``do nothing'' case, it suffices to check that no neighbor of $v$ gets state ``after'' before $v$ has gotten state ``on.''
\end{proof}

Note that doing this operation only makes sense if $v\in B$. Otherwise, it will yield a partial solution that, although valid, can never be extended to a complete solution to the problem since $v$ will never be able to move again. 

\begin{lemma}
   Let $\mathcal{C}'$ be a characteristic of a partial solution $\Pi'$ with respect to $G_{X'}$ that contains the removal of the token on the canonical path of $v$. Then, the characteristic $\mathcal{C}$ and the legality of a partial solution $\Pi$ with respect to $G_{X}$ obtained from $\Pi'$ by replacing the removal with a jump to $v$ and adding a subsequent removal of $v$ depends only on, and can be computed from, $\mathcal{C}'$.
\end{lemma}

\begin{proof}
Again, the only missing information is the state of $v$ itself, which we can compute as we did before, with the exception that at any point after $v$ has changed to ``on,'' we can insert a configuration characteristic into the sequence in which it changes to ``after'' (and accordingly, change its state in all configuration characteristics occuring later on in the sequence to ``after'' as well).

To ensure that the newly created configuration is legal, we must (in addition to the checks previously described) ensure that the state of $v$ does not change to ``after'' while the state of one or more of its neighbors is still ``before.''
\end{proof}

Note that in the characteristic of $\Pi'$, two sucessive configuration characteristics can correspond to a sequence of multiple moves resulting in token configurations having the same configuration characteristic. It does not matter at what exact point in this sequence we insert the jump to $v$, since this results in the same characteristic for $\Pi$. 

This allows us to compute the dynamic programming table after $v$ is introduced: we go over all solution characteristics in the previous table, try all ways of extending them as described above, which yields new solution characteristics together with the number of moves required to reach them (for accounting purposes, we can charge the jump to the addition of a removal). The same solution characteristic may be created in several ways (i.e., from applying different operations to different solution characteristics), and in this case, in the final table, we store the lowest number of moves required for any of them.

The optimality of the solution is guaranteed as follows: if $\Pi$ is an optimal partial solution for $G_X$ with characteristic $\mathcal{C}$, then $\Pi'$, which is obtained by replacing the jump to $v$ (if it is present) with a removal, is optimal for $G_{X'}$ with characteristic $\mathcal{C'}$. Otherwise, if some partial solution $\Pi''$ for $G_{X'}$ would also have characteristic $\mathcal{C'}$ but use fewer moves, it could be extended in the same way as $\Pi'$ can be extended to obtain a partial solution for $G_X$ with the same characteristic as $\Pi$, contradicting the optimality of $\Pi$. 

\paragraph{Forget.} The forget case can be implemented using a simple projection of the stored information to the neighborhood of $X$, merging configuration characteristics that occur repeatedly and taking the minimum among partial solutions that end up getting the same characteristic when projected. Since $X\subseteq X'$, no vertex of $X$ that was not yet reachable in $X'$ becomes reachable due to this projection, so all the partial solutions remain valid.

\paragraph{Final remarks.} Note that since a solution to the token jumping reconfiguration problem is also a partial solution for the final bag of the decomposition, we can read off the desired answer from the dynamic programming table for the final bag.

The table for the starting bag can be computed using brute force. In each table, there are at most $(3^{2k^2+2k}+1)!$ distinct solution characteristics, and we can compute the possible new characteristics that can be reached from each in time $\mathcal{O}(3^{2k^2+2k}! n^{\mathcal{O}(1)})$, since the path decomposition has $\mathcal{O}(n)$ bags, this results in an $\mathcal{O}((3^{2k^2+2k}+2)! n^{\mathcal{O}(1)})$-time algorithm, which is fixed-parameter tractable.
\end{proof}

\section{No polynomial kernel for parameter $k$}
We start with some definitions and a refresher on the machinery required for proving kernelization lower bounds. We use the framework made possible by the work of Drucker~\cite{DBLP:journals/siamcomp/Drucker15}, Bodlaender et al.~\cite{DBLP:journals/jcss/BodlaenderDFH09, DBLP:journals/siamdm/BodlaenderJK14}, Dell and van Melkebeek~\cite{DBLP:journals/jacm/DellM14}, and Fortnow and Santhanam~\cite{DBLP:journals/jcss/FortnowS11}.

\begin{definition}
An equivalence relation $\mathcal{R}$ on $\Sigma^*$ is called a \emph{polynomial equivalence
relation} if the following two conditions hold.
\begin{enumerate}
    \item There is an algorithm that given two strings $x, y \in \Sigma^*$ decides whether $x$ and $y$ belong to the same equivalence class in time polynomial in $|x| + |y|$.
    \item For any finite set $S \subseteq \Sigma^*$ the equivalence relation $\mathcal{R}$ partitions the elements of $S$ into at most $(\max_{x \in S}{|x|})^{\mathcal{O}(1)}$ classes.
\end{enumerate}
\end{definition}

\begin{definition}
Let $L \subseteq \Sigma^*$ be a set and let $Q \subseteq \Sigma^* \times \mathbb{N}$ be a parameterized problem. We say that $L$ \emph{and-cross-composes} into $Q$ if there is a polynomial equivalence relation $\mathcal{R}$ and an algorithm that, given $r$ strings $x_1, x_2, \dots, x_r$  
belonging to the same equivalence class of $\mathcal{R}$, computes an instance $(x^*, \kappa^*) \in \Sigma^* \times \mathbb{N}$ in time polynomial in $\sum_{i = 1}^{r}{|x_i|}$ such that:
\begin{enumerate}
    \item $(x^*, \kappa^*) \in Q$ if and only if $x_i \in L$ for every $1 \leq i \leq r$, and
    \item $\kappa^*$ is bounded by a polynomial in $\max_{i = 1}^{r}{|x_i|} + \log r$.
\end{enumerate}
\end{definition}

\begin{theorem}
If some set $L \subseteq \Sigma^*$ is $\textsf{NP}$-hard under Karp reductions and $L$ and-cross-composes into the parameterized problem $Q$, then there is no polynomial kernel for $Q$ unless $\textsf{NP} \subseteq \textsf{coNP/poly}$.
\end{theorem}

We are now ready to prove the main result of this section:

\begin{theorem}
There exists an and-cross-composition from \textsc{Vertex Cover} into
\textsc{TJ Minimum Separator Reconfiguration}, parameterized by the minimum size $k$ of an \stsep.  Consequently, when parameterized by $k$, \textsc{TJ Minimum Separator Reconfiguration} does not admit a polynomial kernel unless $\textsf{NP} \subseteq \textsf{coNP/poly}$.  
\end{theorem}

\begin{proof}
By choosing an appropriate polynomial
equivalence relation $\mathcal{R}$, we may assume that we are given a family of $r$ \textsc{Vertex Cover} instances $(G_i, \kappa_i)$, where $|V(G_i)| = \mu$ and $\kappa_i = \kappa$ for all $i$. Moreover, we assume, without loss of generality, that each instance of \textsc{Vertex Cover} was obtained by reducing an instance of \textsc{$3$-SAT} having $N$ variables and $M$ clauses (using the standard reduction that converts variables to edges and clauses to triangles). This implies that $\kappa = N + 2M$ and no graph $G_i$ admits a vertex cover of size less than $\kappa$. 

We now proceed to the construction of our \textsc{Minimum Separator Reconfiguration} instance $(H,s,t,A,B)$, where $k = \mu + 1$ will be the size of a minimum separator. The graph $H$ consists of $k$ pairwise internally vertex-disjoint paths from $s$ to $t$. The paths, i.e., $P_1, P_2, \dots, P_\mu, P_k$, consist of $4r + 1$ vertices each (including $s$ and $t$) and correspond to the canonical paths. Recall that we let $u_{i,1}, \dots, u_{i, L(i)}$ denote the vertices on the canonical path $P_i$ in the order in which they appear on it, with $u_{i,1} = s$ and $u_{i, L(i)} = t$. 

The $k$ vertices $\{u_{1,2}, u_{2,2}, \ldots, u_{k,2}\}$ (the vertices adjacent to $s$) correspond to the starting configuration $A$ and the vertices $\{u_{1,L(1) - 1}, u_{2,L(2) - 1}, \ldots, u_{k,L(k) - 1}\}$ (the $k$ vertices adjacent to $t$) correspond to the target configuration $B$. The $\mu$ vertices that come after $A$ (excluding the vertex $u_{k,3}$ in $P_k$) correspond to the vertices of $G_1$ and we use them to embed the edges of $G_1$. Note that the vertex $u_{k,3}$ will remain a vertex of degree two that joins its two neighbors on the canonical path $P_k$. We now proceed with the description of the synchronization gadget that will appear between every pair of embeddings of graphs $G_i$ and $G_{i+1}$. The first synchronization gadget consists of vertices $\{u_{1,5}, u_{2,5}, \ldots, u_{\mu,5}\} \cup \{u_{k,2}, u_{k,4}, u_{k,6}\}$. We add an edge from every vertex in $\{u_{1,5}, u_{2,5}, \ldots, u_{\mu,5}\}$ to every vertex in $\{u_{k,2}, u_{k,4}, u_{k,6}\}$. Note that the vertices $u_{1,4}$ to $u_{\mu,4}$ and $u_{1,6}$ to $u_{\mu,6}$ will remain vertices of degree two. We let $q = 7$. 

We proceed with the construction as follows. For each graph $G_i$, we embed the edges of the graph over the vertices $\{u_{1,q}, u_{2,q}, \ldots, u_{\mu,q}\}$. We then add a synchronization gadget (except when $i = r$) using vertices $\{u_{1,q+2}, \ldots, u_{\mu,q+2}\} \cup \{u_{k,q-1}, u_{k,q+1}, u_{k,q+3}\}$. We add edges from every vertex in $\{u_{1,q+2}, \ldots, u_{\mu,q+2}\}$ to every vertex in $\{u_{k,q-1}, u_{k,q+1}, u_{k,q+3}\}$. We then set $q = q + 4$ and proceed to the next graph $G_{i+1}$. For the final synchronization gadget which preceeds $G_r$, we add edge between the vertex $u_{k,q+3}$ and every vertex in $B \setminus V(P_k)$. This completes the construction of $H$. We ask for a reconfiguration sequence of length $\ell = r(k + \kappa)$ (see Figure~\ref{fig:cross} for more details). 

\begin{figure}[h!]
  \begin{center}
      \pgfdeclarelayer{background}
      \pgfsetlayers{background,main}
              
    \begin{tikzpicture}[decoration={snake,pre length=1mm,post length=1mm,segment length=4mm}]
      \node[draw=black,fill=white,circle,inner sep = 2pt] (s) at (0,0) {};
      \node[above=.1cm of s] {$s$};
      \node[draw=black,fill=white,circle,inner sep = 2pt] (t) at (12,0) {};
      \node[above=.1cm of t] {$t$};

      \foreach \y in {2,1,0,-1,-2}
         {\foreach \x in {1,2,...,10}
            {\draw (\x,\y) -- (\x+1,\y);
            }
          \draw (s) -- (1,\y) (11,\y) -- (t);
         }

      \foreach \y in {2,1,0,-1}
      \draw[mygreen!50!white, thick] (1,-2) edge[bend left=5] (4,\y) (4,\y) edge[bend right=5] (3,-2) (4,\y) edge[bend left=5] (5,-2) (5,-2) edge[bend left=5] (8,\y) (8,\y) edge[bend right=5] (7,-2) (8,\y) edge[bend left=5] (9,-2) (9,-2) edge[bend left=5] (11,\y);
      
      \foreach \y in {2,1,0,-1,-2}
      { \foreach \x in {1,2,...,11}
      \draw (\x,\y) node[draw=black,fill=white,circle,inner sep = 2pt] (u\x\y) {};}
      
      \draw[thick, myorange!70!white] (1,0) ellipse (0.3cm and 2.5cm) (11,0) ellipse (0.3cm and 2.5cm);
      \draw[thick, myblue!50!white] (2,0.5) ellipse (0.3cm and 2cm) (6,0.5) ellipse (0.3cm and 2cm) (10,0.5) ellipse (0.3cm and 2cm);
      
      \draw (1,-2) node[circle, fill=mygreen!50!white, inner sep=2pt] {} (3,-2) node[circle, fill=mygreen!50!white, inner sep=2pt] {} (5,-2) node[circle, fill=mygreen!50!white, inner sep=2pt] {} (7,-2) node[circle, fill=mygreen!50!white, inner sep=2pt] {} (9,-2) node[circle, fill=mygreen!50!white, inner sep=2pt] {};
      
      \foreach \y in {2,1,0,-1}
      \draw (4,\y) node[circle, fill=mygreen!50!white, inner sep=2pt] {} (8,\y) node[circle, fill=mygreen!50!white, inner sep=2pt] {};
    \end{tikzpicture}
  \end{center}
  \caption{An overview of the and-cross-composition with $r = 3$ and $k=5$. The orange blobs are $A$ and $B$, the blue blobs represent the graphs $G_i$, and the green vertices and edges represent the synchronization gadgets.}
  \label{fig:cross}
\end{figure}
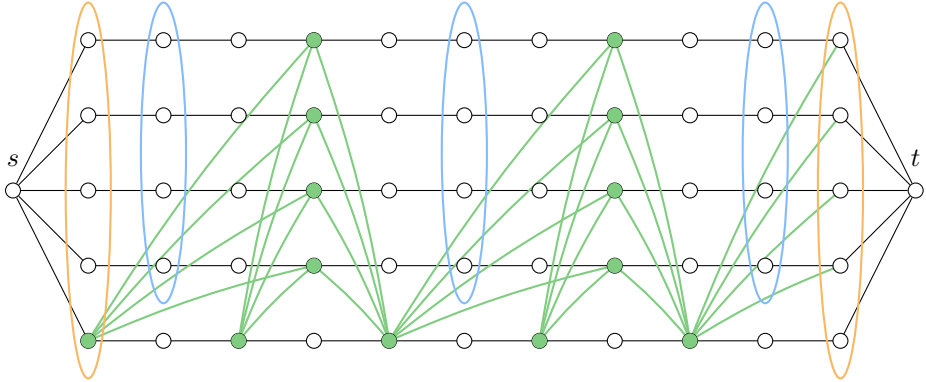 

The intuition behind the construction is similar to the construction showing \textsf{NP}-hardness in Theorem~\ref{thm-jumping-nphard}. The extra synchronization gadgets are required to guarantee the order in which the tokens will cross the vertices corresponding to each graph $G_i$. 

We now prove that we have a reconfiguration sequence of (exactly) $\ell$ jumps from $A$ to $B$ in $H$ if and only if each $G_i$ has a vertex cover $C_i$ of size (exactly) $\kappa$, $1 \leq i \leq r$. 

Assume that each $G_i$ has a vertex cover $C_i$ of size $\kappa$. Then, by Lemma~\ref{lem-hard-2}, we know that we can jump the tokens (over any $G_i$) from the preceeding synchronization gadget (or $A$) to the next synchronization gadget (or $B$) using $k + \kappa$ jumps (with the token on $P_k$ always jumping last to vertex $u_{k,q+3}$, for some $q$). Since we have exactly $r$ graphs to jump over it follows that we can reconfigure $A$ to $B$ via a sequence of $\ell = r(k + \kappa)$ token jumps. 

For the reverse direction, we know, from Lemma~\ref{lem-hard-1}, that whenever we are able to jump all tokens over a graph $G_i$ using at most $k + \kappa$ jumps, $G_i$ must have a vertex cover of size at most $\kappa$. Assume that before any token can jump to either a vertex belonging to some $G_i$ or to some vertex beyond $G_i$, it must be the case that all the tokens have reached a vertex beyond $G_{i-1}$ but before $G_i$. In other words, all the tokens must be on the synchronization gadget between $G_{i-1}$ and $G_i$ before any token can jump onto or beyond $G_i$. Combining this assumption with the fact that no graph $G_i$ can have a vertex cover of size less than $\kappa$, we know that exactly $k + \kappa$ jumps can be used to move the tokens from the gadget preceeding $G_i$ (or from $A$) to the gadget succeeding $G_i$ (or to $B$). Otherwise we need more than $\ell = r(k + \kappa)$ jumps to jump over all $r$ graphs.

Hence, it remains to show that the synchronization gadgets behave as required and force all the tokens to reach the same gadget before jumping onto or over the next graph $G_{i+1}$. First, observe that a synchronization gadget contains $k$ vertices each on a different canonical path and each having two neighbors on some canonical path other than its own. This implies that each such vertex must be occupied by a token at some point in any reconfiguration sequence, i.e., these are unskippable vertices (Lemma~\ref{lemma:unskip}). In addition, given the edges between $\{u_{1,q+2}, \ldots, u_{\mu,q+2}\}$ and $\{u_{k,q-1}\}$ we know that the token on $P_k$ must be the last to reach its vertex $u_{k,q+3}$ on the synchronization gadget; otherwise one of the edges between $\{u_{1,q+2}, \ldots, u_{\mu,q+2}\}$ and $\{u_{k,q-1}\}$ can be used as a crossing edge to obtain an \stpath. So, assume some token on $P_i$ ($i \neq k$) jumps forward beyond its corresponding vertex in the synchronization gadget before the token on $P_k$ gets to $u_{k,q+3}$. We can now construct a path from $t$ to $s$ by simply following $P_k$ to $u_{k,q+3}$, then $u_{k,q+3}$ to $u_{i,q+2}$, and then following the path $P_i$ to $s$; a contradiction. Putting it all together, we know that the tokens on $P_1, P_2, \dots, P_\mu$ will each jump to the unskippable vertex of a synchronization gadget (in some order and staying on their respective canonical paths) before the token on $P_k$ ever reaches its unskippable vertex in the gadget. Moreover, none of the tokens in the first $\mu$ canonical paths can further jump either backward or forward before the token on $P_k$ reaches its unskippable vertex on the gadget. Hence, only when all the tokens have reached a certain synchronization gadget,  we can start jumping over the next graph and exactly $k + \kappa$ jumps can be allocated to move the tokens from the synchronization gadget preceeding $G_i$ to the synchronization gadget succeeding $G_i$ (the same is true for jumping from $A$ to the first synchronization gadget and for jumping from the last synchronization gadget to $B$). 
\end{proof}

\section{Polynomial kernel for parameter $\ell$}
In this section, we consider the \textsc{minimum separator reconfiguration} problem parameterized by the number $\ell$ of allowed token jumps. Formally, given a graph $G$, minimum \stseps $A$ and $B$ of size $k$, and an integer $\ell$, the goal is to determine whether there exists a sequence of at most $\ell$ token jumps transforming $A$ to $B$. 

Recall that, by Lemma~\ref{lemma:preprocesspaths}, we can preprocess the graph so that all vertices are on the canonical paths and the vertices in $A$ and $B$ are adjacent to (at least one of) $s$ or $t$. Similarly, by Lemma~\ref{lemma:degthree}, each vertex in $G$ will have degree at least $3$ and at most $2k$ (with at
most two neighbors on each canonical path) and, by Lemma~\ref{lemma:disjoint}, we know that $A \cap B = \emptyset$. 
We refer to an instance satisfying all of the above as a reduced instance.

\begin{lemma}\label{lem-ell-1}
In a reduced yes-instance where $\ell$ is the parameter, all the following properties must be satisfied:
\begin{itemize}
    \item $A \cap B = \emptyset$ and $|A| = |B| \leq \ell$;
    \item all vertices are on the (at most $\ell$) canonical paths;
    \item vertices in $A$ and $B$ are adjacent to (at least one of) $s$ or $t$; and
    \item each vertex in $G$ has degree at least $3$, degree at most $2\ell$, and at most two neighbors on each canonical path.
\end{itemize} 
\end{lemma}

\begin{proof}
The lemma follows immediately from the fact that if after obtaining a reduced instance we have $|A| = |B| > \ell$ then we have a no-instance; as at least $\ell + 1$ jumps are needed. Consequently, the minimum separator size will be at most $\ell$ and all the remaining properties follow from Lemma~\ref{lemma:preprocesspaths}, Lemma~\ref{lemma:degthree}, and Lemma~\ref{lemma:disjoint}.  
\end{proof}

\begin{lemma}\label{lem-ell-2}
Assume that in a reduced instance one of the canonical paths contains more than $4(\ell + 1)^2 + 4$ vertices. Then, the instance is a no-instance. 
\end{lemma}

\begin{proof}
Let $P$ denote such a canonical path. We claim that at least $\ell + 1$ jumps are required for the token on $P$. We assume otherwise, i.e, that $\ell$ jumps or fewer are enough, and work towards a contradiction. 

First, recall that if a vertex $v$ on a canonical path is adjacent to two distinct vertices on another canonical path, then $v$ can never be jumped over, i.e., $v$ is unskippable (Lemma~\ref{lemma:unskip}).  
We decompose $P$ into $\ell + 1$ subpaths each consisting of at least $4\ell + 4$ vertices (excluding $s$, $t$, and the initial and target vertices of $A$ and $B$). For the token on $P$ to reach its final position in at most $\ell$ jumps, it must (at least once) jump over $2\ell + 1$ (consecutive) vertices of $P$ or more (landing on the vertex $2\ell + 2$ away or more). Let us denote those vertices that are jumped over by $Q$. 
Moreover, let $S_i$ denote the \stsep preceeding the jump and let $S_{i+1}$ denote the resulting \stsep after the jump. 

If $Q$ contains a vertex having two distinct neighbors on another canonical path, then the vertex is unskippable and we get a contradiction. Hence, every vertex $v$ of $Q$ (which has degree $3$ or more) can have at most one neighbor on every canonical path $P' \neq P$ (and we know that $v$ must have at least one neighbor not in $P$). Since $S_{i+1}$ is an \stsep, every vertex $v \in Q$ can only be adjacent to vertices in $S_{i+1} \setminus V(P)$; otherwise an \stpath can be easily constructed, contradicting the fact that $S_{i+1}$ is an \stsep. 
Now, given that $|Q| \geq 2\ell + 1$, we know that there exists at least $3$ distinct vertices in $Q$ all having the same neighbor in $S_{i+1} \setminus V(P)$. This contradicts the fact that after our reductions each vertex in $G$ can have at most two distinct neighbors on any canonical path. 
\end{proof}

\begin{theorem}\label{thm:kernelell}
The optimization version of {\sc TJ Minimum Separator Reconfiguration} admits a kernel with $\mathcal{O}(\ell^2)$ vertices when parameterized by $\ell$, the length of a reconfiguration sequence.
\end{theorem}

\begin{proof}
First, note that Lemmas~\ref{lem-ell-1} and~\ref{lem-ell-2} imply a kernel with $\mathcal{O}(\ell^3)$ vertices; a yes-instance consists of at most $\ell$ canonical paths each having $\mathcal{O}(\ell^2)$ vertices. 

To obtain the improved bound, we have to refine our analysis slightly and strengthen the result of Lemma~\ref{lem-ell-2}. We do so as follows. We partition the canonical paths into long paths and short paths, where a path is short if it has at most $4\ell + 4$ vertices and it is long otherwise. It is not hard to see that the number of vertices belonging to short paths is in  $\mathcal{O}(\ell^2)$. 

Now, assume that the total number of vertices on all long canonical paths is more than $2(2\ell + 2)^2 + 2\ell + 2$. Then, we can decompose the paths into subpaths each consisting of $4\ell + 4$ vertices except possibly the last subpath of each long path (we  exclude $s$, $t$, and the initial and target vertices of $A$ and $B$). The number of subpaths with $4\ell + 4$ vertices is therefore at least $\ell + 1$ and at least $\ell + 1$ jumps are required for the tokens on long paths to reach their destinations, which implies a no-instance (the proof being identical to that of Lemma~\ref{lem-ell-2}). Therefore, the total number of vertices belonging to long paths is also in $\mathcal{O}(\ell^2)$, which concludes the proof. 

We can further improve the bounds on the kernel size (for both the number of vertices and the number of edges), by combining the use of unskippable vertices with the use of unskippable edges. Let $H$ be the graph obtained from $G$ by deleting all the edges belonging to the canonical paths, and let $U_G$ denote the set of unskippable vertices in $V(G) \setminus (A \cup B \cup \{s,t\})$. Now, let $M$ be a maximum matching in the graph $\Gamma = H - (U_G \cup A \cup B)$. 
If $|M| > \ell$ we have a trivial no-instance, since at least one endpoint of each edge in $M$ must be occupied by a token at some point, i.e., every edge in $M$ is unskippable. Hence, every vertex of $\Gamma$ having at least one neighbor is either in $V(M)$ or is adjacent to some vertex of $V(M)$ (isolated vertices in $\Gamma$ correspond to those vertices that only have neighbors in $U_G \cup A \cup B$). Since $|U_G \cup A \cup B| \leq 3\ell$ and the maximum degree is $2\ell$, we get, $|V(\Gamma)| \leq |V(M)| + \sum_{v \in V(M)} |N_\Gamma(v)| + 6\ell^2 \leq 2\ell + \sum_{v \in V(M)} (\ell-1) + 6\ell^2 \leq 2\ell + 2\ell(\ell - 1) + 6\ell^2 \leq 2\ell + 8\ell^2$. Since $V(\Gamma) = V(H) \setminus (U_G \cup A \cup B)$ and $|U_G \cup A \cup B| \leq 3\ell$, we have $|V(G)| = |V(H)| = |U_G \cup A \cup B| + |V(\Gamma)| \leq 3\ell + 2\ell + 8\ell^2 = 5\ell + 8\ell^2$. In terms of edges, each edge of $\Gamma$ is incident to a vertex of $M$, so we have that $|E(\Gamma)| \leq |E(M)| + \sum_{v \in V(M)} |N_\Gamma(v)| \leq \ell + 2\ell(\ell - 1)$ and $|E(H)| = E(\Gamma) + \sum_{v \in U_G \cup A \cup B} |N_G(v)| \leq \ell + 2\ell(\ell - 1) + 6\ell^2 \leq 8\ell^2 + \ell$. Hence, since $|V(G)| \leq 2\ell^2 + 3\ell$, we get $|E(G)| \leq |E(H)| + 2(2\ell^2 + 3\ell) \leq 8\ell^2 + \ell + 2(2\ell^2 + 3\ell) \leq 12\ell^2 + 7\ell$.
\end{proof}


\section{Concluding Remarks}

We studied the minimum \stsep reconfiguration problem through several lenses.
First, we considered the token sliding and token jumping reconfiguration models, showing that the reachability question is answerable in polynomial time in both cases.
Afterwards, we considered the task of finding a shortest reconfiguration sequence; we proved that it is easy under the first model but \NP-complete under the second. To tackle this hardness, we studied the parameterized complexity of the token jumping version for the natural parameterizations $k$ (the number of tokens) and $\ell$ (the length of the sequence).
In this context, we designed an \FPT\ algorithm for the $k$ parameterization, a quadratic kernel when parameterized by $\ell$, and showed that no polynomial kernel exists for $k$ unless $\NP \subseteq \coNP/\textsf{poly}$.

In terms of future work on minimum \stsep reconfiguration itself, we are interested in understanding shortest token jumping sequences for some graph classes, in particular for planar graphs.
Other possibilities include the study of structural parameterizations, such as treewidth, feedback edge set, and vertex deletion distance metrics (e.g. distance to cluster, clique, etc). In a different spirit, studying the connectivity of the minimum \stsep reconfiguration graph might yield different insights than the ones presented in this paper.
Beyond minimum separators, the work of Gomes, Nogueira, and dos Santos~\cite{separator-recon} investigated arbitrary \stsep reconfiguration, but no research has yet been done on bounded size separators.
As such, we believe work on the complexity of reconfiguration of minimum + $r$ \stseps might be of independent interest.

\printbibliography

\newpage
\appendix
\section*{Appendices}
\section{$\textsf{W[1]}$-hardness reduction for parameter $\ell$}\label{appendix-hardness}

\begin{theorem}
The \textsc{Vertex Cover Reconfiguration} problem parameterized by $\ell$, the length of a reconfiguration sequence, is $\textsf{W[1]}$-hard in the token jumping model, even when restricted to bipartite graphs. 
\end{theorem}

\begin{proof}
We reduce the \textsc{Clique} problem to the \textsc{Vertex Cover Reconfiguration} problem on bipartite graphs (under the token jumping model). Given an instance $(G, \kappa)$ of \textsc{Clique}, we create a bipartite graph $H$, with parts $L$ (left) and $R$ (right), as follows. We first add a copy of $V(G)$ to $L$ and another copy to $R$. We use $L_V = \{u_1, \ldots, u_n\}$ and $R_V = \{v
_1, \ldots, v_n\}$ to denote these two sets, respectively. Then, for each edge $e \in E(G)$, we add a vertex $u_e$ in $L$. We use $L_E$ to denote this set. Finally, we add ${\kappa \choose 2} + 1$ vertices to $L$ and ${\kappa \choose 2} + 1$ vertices to $R$, which we denote by $L_K$ and $R_K$, respectively. We proceed by describing the edges of $H$. We add all the edges $\{u_i, v_i\}$, $i \in [n]$, so that the graph induced by $L_V \cup R_V$ is a matching. We add edges between every vertex in $L_K$ and every vertex in $R_K$ such that the graph induced by $L_K \cup R_K$ is a biclique $K_{{\kappa \choose 2} + 1, {\kappa \choose 2} + 1}$. Finally, for every vertex $v_e \in L_E$, where $e = \{v_i, v_j\}$, we add the edges $\{u_e, v_i\}$ and $\{u_e, v_j\}$. This completes the construction of the bipartite graph $H$. 

The corresponding 
\textsc{Vertex Cover Reconfiguration} instance requires a starting vertex cover $A$, which we set as $A = L_V \cup L_E \cup L_K = L$, and a target vertex cover $B$, which we set as $B =  L_V \cup L_E \cup R_K$. In other words, we simply want to switch sides on the biclique so that the tokens occupy the right side instead of the left. Finally, we require a reconfiguration sequence of at most $\ell = 2{\kappa \choose 2} + 2\kappa + 1$ jumps.  

Assume that $G$ has a clique of size $\kappa$ and let $C = \{v_{i_1}, \ldots, v_{i_\kappa}\}$ denote the vertices of that clique. Then, we can reconfigure $A$ to $B$ as follows. 
\begin{enumerate}
    \item Jump the token on $u_{i_1}$ to $v_{i_1}$, then the token on $u_{i_2}$ to $v_{i_2}$, all the way to the token on $u_{i_\kappa}$ to $v_{i_\kappa}$ ($\kappa$ jumps in total). 
    \item Since $C$ is a clique, we now have ${\kappa \choose 2}$ tokens in $L_E$ whose neighborhood in $R$ is full of tokens. So we can jump those tokens to any ${\kappa \choose 2}$ vertices of $R_K$ (${\kappa \choose 2}$ jumps in total). 
    \item Next, we jump any token from $L_K$ to the remaining non-occupied vertex from $R_K$ ($1$ jump in total). 
    \item Now, the remaining tokens in $L_K$ will jump to $L_E$ to replace the lost tokens (${\kappa \choose 2}$ jumps in total). 
    \item Finally, the $\kappa$ tokens in $R_V$ will jump back to their corresponding vertex in $L_V$ ($\kappa$ jumps in total). 
\end{enumerate}

For the reverse direction, assume that we have a sequence of at most $\ell$ jumps transforming vertex cover $A$ to vertex cover $B$. We can assume, without loss of generality, that the tokens in $L_V \cup R_V$ will each remain on their corresponding matching edge and no token from $V(H) \setminus (L_V \cup R_V)$ will ever jump to some vertex in $L_V \cup R_V$; otherwise, if some token jumps to a matching edge so that the other token jumps away, we can simply switch the roles of both tokens and obtain an equivalent sequence. In other words, we can assume that every vertex cover in a reconfiguration sequence will have exactly (the same) $n$ tokens in $L_V \cup R_V$ (each token never leaving its edge). Now, since $H[L_K \cup R_K]$ is a biclique $K_{{\kappa \choose 2} + 1, {\kappa \choose 2} + 1}$, we know that any reconfiguration sequence from $A$ to $B$ must contain some vertex cover $S$ such that $S \cap (L_K \cup R_K) = 2{\kappa \choose 2} + 1$, with $S \cap L_K = {\kappa \choose 2} + 1$ and $S \cap R_K = {\kappa \choose 2}$. From our previous observation, these additional ${\kappa \choose 2}$ tokens must come from $L_E$. Hence, for any sequence of jumps transforming $A$ to $B$, at least $2{\kappa \choose 2} + 1$ jumps are required  for tokens not in $L_V \cup R_V$; at least ${\kappa \choose 2}$ jumps from $A$ to $S$ and at least ${\kappa \choose 2} +  1$ jumps from $S$ to $B$. Given that $\ell = 2{\kappa \choose 2} + 2\kappa + 1$, we know that at most $\kappa$ tokens from $L_V \cup R_V$ can move (and each move will cost two jumps). However, given that $S$ must exist, moving less than $\kappa$ tokens in $L_V \cup R_V$ is impossible since we can never free ${\kappa \choose 2}$ tokens in $L_E$; the neighborhood of any ${\kappa \choose 2}$ vertices in $L_E$ has size at least $\kappa$. Therefore, the corresponding $\kappa$ vertices from $R_V$ must form a clique in $G$, which completes the proof. 
\end{proof}

\begin{theorem}
If we allow multiple tokens to occupy the same vertex then the \textsc{Vertex Cover Reconfiguration} problem parameterized by $\ell$, the length of a reconfiguration sequence, is $\textsf{W[1]}$-hard in the token sliding, even when restricted to bipartite graphs. 
\end{theorem}

\begin{proof}
We give a reduction from the \textsc{Multicolored Clique} problem. Given an instance $(G, \kappa)$ of \textsc{Multicolored Clique}, we assume, without loss of generality, that $V(G)$ is partitioned into $\kappa$ independent sets $V_1$, $V_2$, $\ldots$, $V_\kappa$, each consisting of $n$ vertices.  We create a bipartite graph $H$, with parts $L$ (left) and $R$ (right), as follows. We first add a copy of each $V_i$, $i \in [\kappa]$, to $L$ and another copy to $R$. We use $L_{V_i} = \{u^i_1, \ldots, u^i_n\}$ and $R_{V_i} = \{v^i
_1, \ldots, v^i_n\}$ to denote these two sets, respectively. We let $E_{i,j}$, $i < j$, denote the edges of $G$ connecting vertices in $V_i$ and $V_j$. Now, for each edge $e \in E_{i,j}$, we add a vertex $v_e$ in $L_{E_{i,j}}$, where $L_{E_{i,j}} \subset L$. Finally, we add ${\kappa \choose 2} + 1$ vertices to $L$ and ${\kappa \choose 2} + 1$ vertices to $R$, which we denote by $L_K = \{y_1, \ldots, y_{{\kappa \choose 2} + 1}\}$ and $R_K = \{z_1, \ldots, z_{{\kappa \choose 2} + 1}\}$, respectively. We proceed by describing the edges of $H$. We add all the edges $\{u^i_j, v^i_j\}$, $i \in [\kappa]$ and $j \in [n]$, so that the graph induced by $L_{V_i} \cup R_{V_i}$ is a matching. We add edges between every vertex in $L_K$ and every vertex in $R_K$ such that the graph induced by $L_K \cup R_K$ is a biclique $K_{{\kappa \choose 2} + 1, {\kappa \choose 2} + 1}$. For every vertex $v_e \in L_{E_{i,j}}$, where $e = \{v^i_p, v^j_q\}$, we add the edges $\{v_e, v^i_p\}$ and $\{v_e, v^j_q\}$. We order the sets $L_{E_{i,j}}$ arbitrarily so that we can refer to them as $L_{E_{1}}$, $L_{E_{2}}$, $\ldots$, $L_{E_{{\kappa \choose 2}}}$. Finally, we add all the edges between $z_i$ and (every vertex in) $L_{E_{i}}$, $i \in [{\kappa \choose 2}]$.
This completes the construction of the bipartite graph $H$. We let $L_V = \bigcup_{i \in [\kappa]}{{L_{V_i}}}$, $R_V = \bigcup_{i \in [\kappa]}{{L_{R_i}}}$, and $L_E = \bigcup_{i < j}{L_{E_{i,j}}}$. 

The corresponding 
\textsc{Vertex Cover Reconfiguration} instance requires a starting vertex cover $A$, which we set as $A = L_V \cup L_E \cup L_K$, and a target vertex cover $B$, which we set as $B =  L_V \cup L_E \cup R_K$. Finally, we require a reconfiguration sequence of at most $\ell = 3{\kappa \choose 2} + 2\kappa + 1$ slides.  

Assume that $G$ has a multicolored clique of size $\kappa$ and let $C = \{v^1_{j_1}, \ldots, v^\kappa_{j_\kappa}\}$ denote the vertices of that clique. Then, we can reconfigure $A$ to $B$ as follows. 
\begin{enumerate}
    \item Slide the token on $u^1_{j_1}$ to $v^1_{j_1}$, then the token on $u^2_{j_2}$ to $v^2_{j_2}$, all the way to the token on $u^\kappa_{j_\kappa}$ to $v^\kappa_{j_\kappa}$ ($\kappa$ slides in total). 
    \item Since $C$ is a multicolored clique, we now have ${\kappa \choose 2}$ tokens, one in each set $L_{E_{i}}$, $i \in [{\kappa \choose 2}]$, whose neighborhood in $R$ is full of tokens. So we can slide those tokens to the vertices $\{z_1, \ldots, z_{{\kappa \choose 2}}\}$ (${\kappa \choose 2}$ slides in total). 
    \item Next, we slide any token from $L_K$ to the remaining non-occupied vertex $z_{{\kappa \choose 2} + 1}$ from $R_K$ ($1$ slide in total). 
    \item Now, the remaining tokens in $L_K$ will slide to $R_K$, to $\{z_1, \ldots, z_{{\kappa \choose 2}}\}$, then to $L_E$ to replace the lost tokens ($2{\kappa \choose 2}$ slides in total). 
    \item Finally, the $\kappa$ tokens in $R_V$ will slide back to their corresponding vertex in $L_V$ ($\kappa$ slides in total). 
\end{enumerate}

For the reverse direction, assume that we have a sequence of at most $\ell$ slides transforming vertex cover $A$ to vertex cover $B$. We can assume, without loss of generality, that the tokens in $L_V \cup R_V$ will each remain on their corresponding matching edge; otherwise, if some token from outside the matching slides to a matching edge so that the matching token slides away, we can simply switch the roles of both tokens and obtain an equivalent sequence. In other words, we can assume that every vertex cover in a reconfiguration sequence will have at least $\kappa n$ tokens in $L_V \cup R_V$ (each token never leaving its edge). Now, since $H[L_K \cup R_K]$ is a $K_{{\kappa \choose 2} + 1, {\kappa \choose 2} + 1}$ biclique, we know that any reconfiguration sequence from $A$ to $B$ must contain some vertex cover $S$ such that $S \cap (L_K \cup R_K) = 2{\kappa \choose 2} + 1$, with $S \cap L_K = {\kappa \choose 2} + 1$ and $S \cap R_K = {\kappa \choose 2}$. From our previous observation, these additional ${\kappa \choose 2}$ tokens must come from $L_E$. Hence, for any sequence of slides transforming $A$ to $B$, at least $3{\kappa \choose 2} + 1$ slides are required  for tokens not in $L_V \cup R_V$; at least ${\kappa \choose 2}$ slides from $A$ to $S$ and at least $2{\kappa \choose 2} +  1$ slides from $S$ to $B$. Given that $\ell = 3{\kappa \choose 2} + 2\kappa + 1$, we know that at most $\kappa$ tokens can slide to $R_V$ (as each token will have to slide back). However, given that $S$ must exist, sliding less than $\kappa$ tokens to $R_V$ is impossible since we can never free ${\kappa \choose 2}$ tokens in $L_E$; the neighborhood of any ${\kappa \choose 2}$ vertices in $L_E$ has size at least $\kappa$. Therefore, to free ${\kappa \choose 2}$ tokens in $L_E$ by sliding $\kappa$ tokens to $R_V$, the corresponding vertices (and edges) must form a multicolored clique in $G$, which completes the proof. 
\end{proof}

\begin{theorem}
If we forbid multiple tokens to occupy the same vertex then the \textsc{Vertex Cover Reconfiguration} problem parameterized by $\ell$, the length of a reconfiguration sequence, is $\textsf{W[1]}$-hard in the token sliding, even when restricted to bipartite graphs. 
\end{theorem}

\begin{proof}
We give a reduction from the \textsc{Multicolored Biclique} problem. An instance of the \textsc{Multicolored Biclique} problem consists of a bipartite graph $G = (X \cup Y, E)$ and a parameter $\kappa$, where $X$ and $Y$ are partitioned into $\kappa$ sets (each of size $n$), i.e., $X = X_1 \cup \ldots X_\kappa$, and $Y = Y_1 \cup \ldots \cup Y_\kappa$. We assume, without loss of generality, that each vertex in $X_i$ ($Y_i$), $i \in [\kappa]$, has at least two neighbors and at least two non-neighbors in each set of $Y_j$ ($X_j$), $j \in [\kappa]$. The goal is to decide if $G$ contains a subgraph isomorphic to the biclique $K_{\kappa,\kappa}$ containing exactly one vertex from each $X_i$ and exactly one vertex from each $Y_i$, $i \in [\kappa]$. We call such a biclique a multicolored biclique. 

Given an instance $(G = (X \cup Y, E), \kappa)$ of \textsc{Multicolored Biclique}, we create a bipartite graph $H$, with parts $L$ (left) and $R$ (right), as follows. We first add a copy of each $X_i$, $i \in [\kappa]$, to $L$ which we denote by $L_{X_i} = \{x^i_1, \ldots, x^i_n\}$. Similarly, we add a copy of each $Y_i$, $i \in [\kappa]$, to $R$ which we denote by $R_{Y_i} = \{y^i_1, \ldots, y^i_n\}$. We use $L_X$ and $R_Y$ to denote the union of the sets on each side, respectively. We add $\kappa + 1$ new vertices to $L$ and $\kappa + 1$ new vertices to $R$, which we denote by $L_K = \{u_1, \ldots, u_{\kappa + 1}\}$ and $R_K = \{v_1, \ldots, v_{\kappa + 1}\}$, respectively. Finally, we add $\kappa + 1$ new vertices to $L$ and $\kappa + 1$ new vertices to $R$, which we denote by $L_Q = \{w_1, \ldots, w_{\kappa + 1}\}$ and $R_Q = \{z_1, \ldots, z_{\kappa + 1}\}$, respectively. We proceed by describing the edges of $H$. We add a copy of each non-edge in $G$ between the corresponding vertices in $L_X \cup R_Y$, i.e., we add an edge between a vertex in $L_X$ and a vertex in $R_Y$ whenever the corresponding vertices are not adjacent in $G$. We add edges between every vertex in $L_K$ and every vertex in $R_K$ as well as between every vertex in $L_Q$ and every vertex in $R_Q$ such that the graphs induced by $L_K \cup R_K$ and $L_Q \cup R_Q$ are bicliques $K_{\kappa + 1, \kappa + 1}$. Finally, we connect every vertex $u_i$ with every vertex in $R_{Y_i}$, every vertex $v_i$ with every vertex in $L_{X_i}$, every vertex $w_i$ with every vertex in $R_{Y_i}$, and every vertex $z_i$ with every vertex in $L_{X_i}$, $i \in [\kappa]$. This completes the construction of $H$, which is clearly a  bipartite graph. 

The corresponding 
\textsc{Vertex Cover Reconfiguration} instance requires a starting vertex cover $A$, which we set as $A = L_X \cup R_Y \cup L_K \cup R_Q$, and a target vertex cover $B$, which we set as $B =  L_X \cup R_Y \cup R_K \cup L_Q$. Finally, we require a reconfiguration sequence of at most $\ell = 4\kappa + 2$ slides.  

Assume that $G$ has a multicolored biclique and let $\{x^1_{i_1}, \ldots, x^\kappa_{i_\kappa}, y^1_{j_1}, \ldots, y^\kappa_{j_\kappa}\}$ denote the vertices of that biclique. Then, we can reconfigure $A$ to $B$ as follows. 
\begin{enumerate}
    \item Slide the token on $x^1_{i_1}$ to $v_{1} \in R_K$, then the token on $x^2_{i_2}$ to $v_{2} \in R_K$, all the way to the token on $x^\kappa_{i_\kappa}$ to $v_{\kappa} \in R_K$ ($\kappa$ slides in total). 
    \item Slide the token on $y^1_{j_1}$ to $w_{1} \in L_Q$, then the token on $y^2_{j_2}$ to $w_{2} \in L_Q$, all the way to the token on $y^\kappa_{j_\kappa}$ to $w_{\kappa} \in L_Q$ ($\kappa$ slides in total). Note that this is possible since the vertices $\{x^1_{i_1}, \ldots, x^\kappa_{i_\kappa}\}$ and the vertices $\{y^1_{j_1}, \ldots, y^\kappa_{j_\kappa}\}$ are pairwise non-adjacent in $H$. 
    \item Next, we slide the token on $u_{\kappa + 1} \in L_K$ to the remaining non-occupied vertex $v_{\kappa + 1} \in R_K$ and we slide the token on $z_{\kappa + 1} \in R_Q$ to the remaining non-occupied vertex $w_{\kappa + 1} \in L_Q$ ($2$ slides in total). 
    \item Now, the remaining tokens in $L_K$ will slide to $R_Y$ to replace the lost tokens ($\kappa$ slides in total). 
    \item Finally, the remaining tokens in $R_Q$ will slide to $L_X$ to replace the lost tokens ($\kappa$ slides in total).
\end{enumerate}

For the reverse direction, assume that we have a sequence of at most $\ell$ slides transforming vertex cover $A$ to vertex cover $B$. Since $H[L_K \cup R_K]$ and $H[L_Q \cup R_Q]$ are $K_{\kappa + 1, \kappa + 1}$ bicliques, we know that any reconfiguration sequence from $A$ to $B$ must contain an  (earliest) vertex cover $S$ followed by som (earliest) vertex cover $T$ ($S$ and $T$ not consecutive but $S$ appears before $T$) such that one of the following is true:
\begin{enumerate}
    \item Either $S \cap (L_K \cup R_K) = 2\kappa + 1$, with $S \cap L_K = \kappa$ and $S \cap R_K = \kappa + 1$, no token will subsequently (after $S$) slide from $R_K$, $T \cap (L_Q \cup R_Q) = 2\kappa + 1$, with $T \cap L_Q = \kappa + 1$ and $T \cap R_Q = \kappa$, and no token will subsequently (after $T$) slide from $L_Q$; or
    \item $S \cap (L_Q \cup R_Q) = 2\kappa + 1$, with $S \cap L_Q = \kappa + 1$ and $S \cap R_Q = \kappa$, no token will subsequently (after $S$) slide from $L_Q$, $T \cap (L_K \cup R_K) = 2\kappa + 1$, with $T \cap L_K = \kappa$ and $T \cap R_K = \kappa + 1$,  and no token will subsequently (after $T$) slide from $R_K$.
\end{enumerate}
We assume, without loss of generality, that $S$ and $T$ are the earliest vertex covers in the sequence such that $S$ and $T$ satisfy the first condition (condition (1)). Note that, since we do not allow multiple tokens to occupy the same vertex, we must have $\kappa$ tokens that slide from $L_X$ to $R_K$ and $\kappa$ tokens that slide from $R_Y$ to $L_Q$. Combined with the fact that $\ell = 4\kappa + 2$ (and that $2$ slides are required for the tokens on $u_{\kappa + 1}$ and $z_{\kappa + 1}$), we know that exactly $\kappa$ tokens will leave $L_X$, exactly $\kappa$ tokens will enter $L_X$, exactly $\kappa$ tokens will leave $R_Y$, and exactly $\kappa$ tokens will enter $R_Y$. 

Let the sequence of slides be chosen (amongst all sequences of at most $\ell$ slides transforming $A$ to $B$) to be a sequence that minimizes the number of slides transforming $A$ to $S$. We claim that, in this case, $S$ must exclude at least one vertex from each $L_{X_i}$, $i \in [\kappa]$.  Assume otherwise. Then, since $S \cap (L_K \cup R_K) = 2\kappa + 1$, every token on $v_{i} \in S \cap R_K$ must come from a vertex in $L_{X_i}$, $i \in [\kappa]$ (recall that we do not allow multiple tokens to occupy the same vertex). Hence, in a sequence that minimizes the number of slides transforming $A$ to $S$, if $S$ includes all vertices from some $L_{X_i}$ then it must be the case that the sequence contains another vertex cover $S'$, prior to $S$, such that $S'$ is the last vertex cover in the sequence that slides some token to $L_{X_i}$ which is different from the token already on $v_i$ (the slide must be from $R_Y$ to $L_X$). We can assume, without loss of generality, that the sequence transforming $S'$ to $S$ consists only of slides from $L_X$ to $R_K$; the only other possible slides are between $R_Y$ and  $L_Q$ which can be safely deleted contradicting our choice of sequence. Now, consider the sequence that transforms $A$ to $S$ without including the slide resulting in $S'$. It is not hard to see that the sequence remains valid; $R_K \cup R_Q$ cover all the edges between $L_{X_i}$ and $R_K \cup R_Q$ which implies that the slide resulting in $S'$ corresponds to a slide from some vertex $y \in R_Y$ to some vertex $x \in L_X$ such that all other neighbors of $x$ are in $S'$ as well as the vertex cover preceeding $S'$. Therefore, the existence of $S'$ implies that the sequence can be made shorter and we get the required contradiction. Using similar arguments, we can show that $S$ must exclude at most one vertex from each $L_{X_i}$, $i \in [\kappa]$. Consequently, we assume in what follows that $S$ excludes exactly one vertex from each $L_{X_i}$, $i \in [\kappa]$. We denote those vertices in $L_X$ (which are not in $S$) by $W = \{x^1_{j_1}, x^2_{j_2}, \ldots, x^\kappa_{j_\kappa}\}$. 

We now consider the vertex covers $S_{q_1}$, $S_{q_2}$, $\ldots$, $S_{q_\kappa}$, appearing in the sequence (not necessarly consecutively) but in order and all succeeding $S$ (by one or more sets) and such that $S_{q_i}$ is the earliest vertex cover (after $S$) resulting from sliding a token from some vertex $y^{q_i}_j \in R_{Y_{q_i}}$, $q_i \in [\kappa]$ and $j \in [n]$, to vertex $w_{q_i} \in L_Q$. We claim that, for each $S_{q_i}$, 
we have $S_{q_i} \cap W = \emptyset$. Hence, after $S_{q_\kappa}$ we have a set that excludes exactly one vertex from each $L_{X_i}$ and each $R_{Y_i}$, $i \in [\kappa]$, and we are done since we have $2\kappa$ vertices corresponding to a multicolored biclique in $G$. Consider the first set where the claim is not true and denote this set by $S_{q'}$. That is, we have $S_{q'} \cap W \neq \emptyset$. This implies that some token slides occured from $R_Y$ to $W$ (tokens in $R_K$ are assumed to be fixed and tokens in $R_Q$ cannot yet slide). However, as previously noted, no more than $\kappa$ tokens can slide out of $R_Y$ and therefore $S_{q'}$ cannot exist. 
\end{proof}

\end{document}